\ifx\ESAVer\undefined
   \documentclass[12pt]{article}
   \newcommand{\InESAVer}[1]{}%
   \newcommand{\InNotESAVer}[1]{#1}%
   \newcommand{\SeeFullVer}{}%
\else
   
   \documentclass{styles/llncs}
   \newcommand{\InESAVer}[1]{#1}%
   \newcommand{\InNotESAVer}[1]{}%
   \newcommand{\SeeFullVer}{ See the full version of the paper for details.}%
\fi

\InNotESAVer{%
   \usepackage[cm]{fullpage}%
}%

\InESAVer{\usepackage{amsthm}}%
\usepackage{amssymb}
\usepackage{bbm}
\usepackage[nocompress]{cite}
\usepackage{enumerate}
\usepackage{euscript}
\usepackage{graphicx}
\usepackage{xcolor}%
\usepackage{hyperref}%
    	\hypersetup{%
     		 breaklinks,%
      		colorlinks=true,%
   		linkcolor=[rgb]{0.45,0.0,0.0},%
   		citecolor=[rgb]{0,0,0.45}
    	}
\usepackage{paralist}
\usepackage{styles/picins}
\usepackage[caption=false]{subfig}
\usepackage{xspace}

\usepackage{tikz}
\usetikzlibrary{external}
\InNotESAVer{%
   \usepackage{titlesec}%
   \titlelabel{\thetitle. }%

   \usepackage{amsmath}%
   \usepackage[amsmath,thmmarks]{ntheorem}
   \theoremseparator{.}
}

\def\line#1{\overleftrightarrow{#1}}
\newcommand{\remove}[1]{}

\newcommand{\unipoint}{unipoint\xspace}

\newcommand{\multipoint}{multipoint\xspace}
\newcommand{\Multipoint}{Multipoint\xspace}
\newcommand{\MultipointSet}{\{(\SiX{1},\PriX{1}),\ldots,(\SiX{m},\PriX{m})\}}

\newcommand{\UPntSet}{\EuScript{P}}

\newcommand{\pnt}{p}
\newcommand{\query}{q}        
\newcommand{\Sites}{P}

\definecolor{blue25}{rgb}{0,0,0.7}
\newcommand{\emphic}[2]{%
     		\textcolor{blue25}{%
         		\textbf{\emph{#1}}}%
         		\index{#2}}
\newcommand{\emphi}[1]{\emphic{#1}{#1}}

\newcommand{\mprChar}{\mu} 
\newcommand{\mprX}[1]{\mprChar\pth{#1}} 

\newcommand{\hmprX}[1]{\widehat{\mprChar}\pth{#1}} 

\newcommand{\pr}{\gamma}
\newcommand{\PR}{\ensuremath{\Gamma}}

\newcommand{\tukey}{\tau}
\newcommand{\Prob}[1]{\ensuremath{\mathrm{Pr}}\!\left[\,#1\, \right]}
\newcommand{\Probsmall}[1]{\ensuremath{\mathrm{Pr}\,}\!\big[\,#1\, \big]}
\newcommand{\some}{\beta}
\newcommand{\SetB}{B}
\renewcommand{\Re}{{\rm I\!\hspace{-0.025em} R}}

\newcommand{\brc}[1]{\left\{ {#1} \right\}}
\newcommand{\pth}[2][\!]{#1\left({#2}\right)}

\newcommand{\w}{\pr}

\newcommand{\Q}{\EuScript{Q}}

\newcommand{\etal}{\textsl{et~al.\ }}

\newcommand\fixedvalue{\sigma}

\InESAVer{%
   \newcommand{\mparagraph}[1]{\medskip \noindent{\emph{\textbf{{#1}}}}   }
}

\InNotESAVer{
   \newcommand{\mparagraph}[1]{\paragraph{#1}}
}

\newcommand{\TT}{\EuScript{T}}
\newcommand{\RSample}{R}

\newcommand{\tDepth}{t}
\newcommand{\simplex}{S}
\newcommand{\tradeoffPara}{t}

\newcommand{\UpperHull}{\EuScript{U}}
\newcommand{\TauLevel}{\Lambda}

\newcommand{\PntSetA}{\Q}

\newcommand{\CHChar}{\text{\textsc{ch}}}%
\newcommand{\CHX}[2][\!]{\CHChar\pth[#1]{#2}}
\newcommand{\CHSomeX}[2][\!]{\CHChar_\some\pth[#1]{#2}}
\newcommand{\ConvexTauHull}[1]{\CHChar_\some(#1)}

\newcommand{\ceil}[1]{\left\lceil {#1} \right\rceil}

\newcommand{\atgen}{\symbol{'100}}%

\newcommand{\SarielThanks}[1]{\thanks{%
      Department of Computer Science; %
      University of Illinois; %
      201 N. Goodwin Avenue; %
      Urbana, IL, 61801, USA;%
      {\tt sariel\atgen{}uiuc.edu}; %
      {\tt \url{http://sarielhp.org/} }%
      . %
      #1}}

\newcommand{\PankajThanks}[1]{\thanks{
      Department of Computer Science; %
      Duke University; %
      Durham, NC, 27708, USA; %
      {\tt pankaj\atgen{}cs.duke.edu <http://www.cs.duke.edu/>}; %
      {\tt \url{http://www.cs.duke.edu/~pankaj/ }}.%
   }%
}

\newcommand{\WuzhouThanks}[1]{\thanks{%
      Department of Computer Science; %
      Duke University; %
      Durham, NC, 27708, USA; %
      {\tt wuzhou\atgen{}cs.duke.edu <http://www.cs.duke.edu/>}; %
      {\tt \url{http://www.cs.duke.edu/~wuzhou/ }}.%
   }%
}

\newcommand{\SubhashThanks}[1]{\thanks{%
      Department of Computer Science; %
      University of California; %
      Santa Barbara; %
      CA, 93106, USA; %
      {\tt suri\atgen{}cs.ucsb.edu}; %
      {\tt \url{http://www.cs.ucsb.edu/~suri}}.%
   }}%

\newcommand{\HakanThanks}[1]{\thanks{%
      Department of Computer Science; %
      University of California; %
      Santa Barbara; %
      CA, 93106, USA; %
      {\tt hakan\atgen{}cs.ucsb.edu}; %
      {\tt \url{http://www.cs.ucsb.edu/~hakan}}.%
   }}

\newcommand{\Caratheodory}{Carath{\'e}odory\xspace}
\newcommand{\SimplexSet}{\EuScript{S}}
 
\newcommand{\lemlab}[1]{\label{lemma:#1}}
\newcommand{\lemref}[1]{Lemma~\ref{lemma:#1}}
\newcommand{\theolab}[1]{\label{theorem:#1}}
\newcommand{\theoref}[1]{Theorem~\ref{theorem:#1}}
\newcommand{\cardin}[1]{\left| {#1} \right|}%
\newcommand{\eps}{{\varepsilon}}%

\renewcommand{\th}{th\xspace}%

\newcommand{\probX}[1]{\mprX{#1}}%
\newcommand{\tprobX}[1]{\widetilde{\mprChar}\pth{#1}}%

\newcommand{\Int}{\mathop{\mathrm{int}}}

\newcommand{\vect}[1]{\ray{#1}}

\def\SHOWCOMMENTS{0} 

\def\ONE{1}

\newcommand{\seg}[1]{{#1}}
\newcommand{\Line}[1]{\overleftrightarrow{#1}}
\newcommand{\ray}[1]{\overrightarrow{#1}}
\newcommand{\SFACE}{\ensuremath{P_\alpha}}
\newcommand{\SRIDGE}{\ensuremath{L_\beta}}

\newcommand{\spi}[1]{\ensuremath{p'_{#1}}}

\newcommand{\siX}[1]{\ensuremath{\pnt_{#1}}}

\newcommand{\SiX}[1]{\ensuremath{P_{#1}}}

\newcommand{\sijY}[2]{\ensuremath{p^{#2}_{#1}}}

\newcommand{\ii}{\ensuremath{u}}
\newcommand{\jj}{\ensuremath{v}}

\newcommand{\priX}[1]{\ensuremath{\pr_{#1}}}

\newcommand{\prijY}[2]{\ensuremath{\pr_{#1}^{#2}}}

\newcommand{\PriX}[1]{\ensuremath{\PR_{#1}}}

\newcommand{\st}{\ensuremath{\,\mid\,}}
\newcommand{\mpt}{\ensuremath{q}}

\newcommand{\GroupX}[1]{\ensuremath{G_{#1}}}

\newcommand{\SGroup}{\ensuremath{G}}

\newcommand{\CC}{C}
\newcommand{\CCP}{C'}
\newcommand{\VV}{V}
\newcommand{\VVP}{V'}

\newcommand{\f}{\ensuremath{f}}
\newcommand{\rray}{\ray{r}(\siX{i}',\query')}

\newcommand{\HYP}{H}

   \usepackage[textsize=scriptsize]{todonotes}

\newcommand{\hakan}[1]{%
   \ifx \SHOWCOMMENTS \ONE
   \todo[fancyline,color=white,bordercolor=red,linecolor=gray]{H: #1}
   \fi%
}

\newcommand{\subhash}[1]{%
   \ifx \SHOWCOMMENTS \ONE
   \todo[fancyline,color=white,bordercolor=blue,linecolor=gray]{SS:
      #1} 
   \fi%
}

\newcommand{\pankaj}[1]{%
   \ifx \SHOWCOMMENTS \ONE
   \todo[fancyline,color=white,bordercolor=black,linecolor=gray]{P:
      #1}%
   \fi%
}

\newcommand{\sariel}[1]{%
   \ifx \SHOWCOMMENTS \ONE
   \todo[fancyline,color=white,bordercolor=yellow,linecolor=gray]{S:
      #1}%
   \fi%
}

\newcommand{\wuzhou}[1]{%
   \ifx \SHOWCOMMENTS \ONE
   \todo[fancyline,color=white,bordercolor=green,linecolor=gray]{W:
      #1}%
   \fi%
}

\newcommand*\circled[1]{\tikz[baseline=(char.base)]{\node[shape=circle,draw,inner sep=1pt] (char) {#1};}}

\newcommand{\inv}{\overline}
\newcommand{\lowd}{\lambda}
\newcommand{\HyperplaneSet}{H}
\newcommand{\Arrangement}{\EuScript{A}}

\newcommand{\seclab}[1]{\label{sec:#1}}
\newcommand{\secref}[1]{Section~\ref{sec:#1}}
\newcommand{\apndlab}[1]{\label{apnd:#1}}
\newcommand{\apndref}[1]{Appendix~\ref{apnd:#1}}
\newcommand{\figlab}[1]{\label{fig:#1}}
\newcommand{\figref}[1]{Figure~\ref{fig:#1}}

\newcommand{\itemlab}[1]{\label{item:#1}}
\newcommand{\itemref}[1]{(\ref{item:#1})}

\newcommand{\thmlab}[1]{{\label{theo:#1}}}
\newcommand{\thmref}[1]{Theorem~\ref{theo:#1}}

\newcommand{\prA}{\ensuremath{\pr(\SetB)}}
\newcommand{\SetBA}{{\SetB \cup \{\query\}}}
\newcommand{\SetBAP}{\SetB' \cup \{\query'\}}

\newcommand{\FaceIndices}{\ensuremath{I_\alpha}}

\newcommand{\altequiv}{=}

\InNotESAVer{%
   \newtheorem{theorem}{Theorem}[section]
   \newtheorem{lemma}[theorem]{Lemma}%
   \newtheorem{corollary}[theorem]{Corollary}

{\theorembodyfont{\rm}%
}

}

\newcommand{\si}[1]{#1}

\InNotESAVer{%
   \theoremstyle{plain}%
   \theoremheaderfont{\sf}%
   \theorembodyfont{\upshape}%
   \theoremseparator{.}%

   \newtheorem*{remark:unnumbered}[theorem]{Remark}%

   \newcommand{\myqedsymbol}{\rule{2mm}{2mm}}

   \theoremheaderfont{\em}%
   \theorembodyfont{\upshape}%
   \theoremstyle{nonumberplain}%
   \theoremseparator{}%
   \theoremsymbol{\myqedsymbol}
   \newtheorem{proof}{Proof:}%
}

\title{Convex Hulls under Uncertainty}

\author{%
   Pankaj K. Agarwal%
   \InNotESAVer{\PankajThanks{}}%
   \InESAVer{\inst{1}}%
   \and%
   Sariel Har-Peled%
   \InNotESAVer{\SarielThanks{}}%
   \InESAVer{\inst{2}}%
   \and %
   Subhash Suri%
   \InNotESAVer{\SubhashThanks{}}%
   \InESAVer{\inst{3}} %
   \and 
   \InESAVer{\newline} 
   Hakan Y{\i}l{}d{\i}z%
   \InNotESAVer{\HakanThanks{}}%
   \InESAVer{\inst{3}}%
   \and%
   Wuzhou Zhang%
   \InNotESAVer{\WuzhouThanks{}}%
   \InESAVer{\inst{1}}%
}

\InESAVer{%
   \institute{%
      Duke University%
      \and%
      University of Illinois, Urbana-Champaign%
      \and%
      University of California, Santa Barbara }
   
   \pagestyle{plain} 
}

\begin{document}

\maketitle

\begin{abstract}
    We study the convex-hull problem in a probabilistic setting,
    motivated by the need to handle data uncertainty inherent in many
    applications, including sensor databases, location-based services
    and computer vision. In our framework, the uncertainty of each
    input site is described by a probability distribution over a
    finite number of possible locations including a \emphi{null}
    location to account for non-existence of the point. Our results
    include both exact and approximation algorithms for computing the
    probability of a query point lying inside the convex hull of the
    input, time-space tradeoffs for the membership queries, a
    connection between Tukey depth and membership queries, as well as
    a new notion of $\some$-hull that may be a useful representation
    of uncertain hulls.
\end{abstract}

%
%

\section{Introduction}

The convex hull of a set of points is a fundamental structure in
mathematics and computational geometry, with wide-ranging applications
in computer graphics, image processing, pattern recognition, robotics,
combinatorics, and statistics. Worst-case optimal as well as
output-sensitive algorithms are known for computing the convex hull;
see the survey \cite{s-chc-97} for an overview of known results.

In many applications, such as sensor databases, location-based
services or computer vision, the location and sometimes even the
existence of the data is uncertain, but statistical information can be
used as a probability distribution guide for data.  This raises the
natural computational question: what is a robust and useful convex
hull representation for such an uncertain input, and how well can we
compute it? We explore this problem under two simple models in which
both the location and the existence (presence) of each point is
described probabilistically, and study basic questions such as what is
the probability of a query point lying inside the convex hull, or what
does the probability distribution of the convex hull over the space
look like.

\mparagraph{Uncertainty models.} %
We focus on two models of uncertainty: unipoint and multipoint.  In
the \emphi{{\unipoint} model}, each input point has a fixed location
but it only exists probabilistically.  Specifically, the input
$\UPntSet$ is a set of pairs $\{(\siX{1},\priX{1}),\ldots,$
$(\siX{n},\priX{n})\}$ where each $\siX{i}$ is a point in $\Re^d$ and
each $\priX{i}$ is a real number in the range $(0,1]$ denoting the
probability of $\siX{i}$'s existence. The existence probabilities of
different points are independent; $P = \brc{p_1, \ldots, p_n}$ denotes
the set of sites in $\UPntSet$.

In the \emphi{\multipoint model}, each point probabilistically exists
at one of multiple possible sites. Specifically, $\UPntSet$ is a set of
pairs
\begin{math}
    \brc{ (\SiX{1},\PriX{1}),\ldots,(\SiX{m},\PriX{m})}
\end{math}
where each $\SiX{i}$ is a set of $n_i$ points and each $\PriX{i}$ is a
set of $n_i$ real values in the range $(0,1]$.  The set
$\SiX{i}=\brc{\sijY{i}{1},\ldots,\sijY{i}{n_i}}$ describes the
possible sites for the $i$\th point of $\UPntSet$ and the set $\PriX{i}
= \brc{\prijY{i}{1},\ldots,\prijY{i}{n_i}}$ describes the associated
probability distribution.  The probabilities $\prijY{i}{j}$ correspond
to disjoint events and therefore sum to at most $1$.  By allowing the
sum to be less than one, this model also accounts for the possibility
of the point not existing (i.e. the \emphi{null} location)---thus, the
multipoint model generalizes the unipoint model. In the multipoint
model, $P = \bigcup_{i = 1}^m P_i$ refers to the set of all sites and
$n =|P|$.

\mparagraph{Our results.} %
The main results of our paper can be summarized as follows.

\begin{compactenum}[(A)]
    \item We show (in \secref{mem}) that the membership probability of
    a query point $\query \in \Re^d$, namely, the probability of
    $\query$ being inside the convex hull of $\UPntSet$, can be
    computed in $O(n\log{n})$ time for $d=2$. For $d \ge 3$, assuming
    the input and the query point are in general position, the
    membership probability can be computed in $O(n^{d})$ time.
    The results hold for both unipoint and multipoint models.
    
    \item Next we describe two algorithms (in \secref{mem_queries}) to
    preprocess $\UPntSet$ into a data structure so that for a query
    point its membership probability in $\UPntSet$ can be answered
    quickly. The first algorithm constructs a \emphi{probability map}
    $\mathbb{M}(\UPntSet)$, a partition of $\Re^d$ into convex cells,
    so that all points in a single cell have the same membership
    probability.  We show that $\mathbb{M}(\UPntSet)$ has size
    $\Theta(n^{d^2})$, and for $d=2$ it can be computed in optimal
    $O(n^4)$ time.  The second one is a sampling-based Monte Carlo
    algorithm for constructing a near-linear-size data structure that
    can approximate the membership probability with high likelihood in
    sublinear time for any fixed dimension.
    
    \item We show (in \secref{tukey}) a connection between the
    membership probability and the Tukey depth, which can be used to
    approximate cells of high membership probabilities.  For $d = 2$,
    this relationship also leads to an efficient data structure.
    
    \item Finally, we introduce the notion of \emphi{$\some$-hull} (in
    \secref{some_hull}) as another approximate representation for
    uncertain convex hulls in the multipoint model: a convex set $C$
    is called \emphi{$\some$-dense} for $\UPntSet$, for $\some \in
    [0,1]$, if $C$ contains at least $\some$ fraction of each
    uncertain point. The $\some$-hull of $\UPntSet$ is the intersection
    of all $\some$-dense sets for $\UPntSet$. We show that for $d = 2$,
    the $\some$-hull of $\UPntSet$ can be computed in $O(n\log^3 n)$
    time.
\end{compactenum}

\mparagraph{Related work.} %
There is extensive and ongoing research in the database community on
uncertain data; see~\cite{drs-pddd-09} for a survey. In the
computational geometry community, the early work relied on
deterministic models for uncertainty (see e.g.~\cite{m-diicg-09}),
but more recently probabilistic models of uncertainty, which are
closer to the models used in statistics and machine learning, have
been explored~\cite{acty-iud-09, aahpy-nnsuu-13, p-sddgs-09,
   kcs-cppop-11, kcs-smstes-11, svy-mlchu-13}.  The convex-hull
problem over uncertain data has received some attention very
recently. Suri \etal ~\cite{svy-mlchu-13} showed that the problem of computing the most likely
convex hull of a point set in the {\multipoint} model is NP-hard. Even
in the {\unipoint} model, the problem is NP-hard for $d \ge 3$. They
also presented an $O(n^3)$-time algorithm for computing the most
likely convex hull under the {\unipoint} model in $\Re^2$. Zhao
\etal~\cite{zyn-pchqt-12} investigated the problem of computing the
probability of each uncertain point lying on the convex hull, where
they aimed to return the set of (uncertain) input points whose
probabilities of being on the convex hull are at least some threshold. J{\o}rgensen \etal~\cite{jlj-gcotp-11} showed that the distribution of properties, such as areas or perimeters, of the convex hull of $\UPntSet $ may have $\Omega(\Pi_{i = 1}^{m} n_{i})$ complexity if all the sites lie on or near a circle.

%
%

\section{Computing the Membership Probability}
\seclab{mem}

For simplicity, we describe our algorithms under the {\unipoint}
model, and then discuss their extension to the {\multipoint} model.
We begin with the 2D case.

\subsection{The two-dimensional case} %
Let $\UPntSet = \brc{(\siX{1},\priX{1}),\ldots,(\siX{n},\priX{n})}$ be
a set of $n$ uncertain points in $\Re^2$ under the {\unipoint}
model. Recall that $\Sites = \brc{\siX{1}, \ldots, \siX{n}}$ is the
set of all sites of $\UPntSet$. 
\hakan{Added this sentence.} For simplicity of description, we assume that the sites are in general position, i.e., no two share coordinates and no three are collinear.
A subset $\SetB \subseteq \Sites$ is the
outcome of a probabilistic experiment with probability
\begin{align*}
    \prA%
    =%
    \prod_{\siX{i} \in \SetB} {\priX{i}} \; \times\; \prod_{ \siX{i}
       \notin \SetB}{\inv{\priX{i}}},
\end{align*}
where $\inv{\priX{i}}$ is the complementary probability $1 -
\priX{i}$.
%
By definition, for a point $\query$, the \emphi{probability} of $\query$ to lie in the convex-hull of $\SetB$ \hakan{Replaced ',' with '$\st$' in the summation below.} is
\begin{align*}
    \displaystyle%
    \mprX{\query}%
    =%
    \sum_{{\SetB \subseteq \Sites \st \,
          \query\,\in\,\CHX[]{\SetB}}}{\prA},
\end{align*}
where $\CHX{\SetB}$ is the convex hull of $\SetB$.  This unfortunately
involves an exponential number of terms.  However, observe that for a
subset $\SetB \subseteq \Sites$, the point $\query$ is \emphi{outside}
$\CHX{\SetB}$, if and only if $\query$ is a vertex of the convex hull
$\CHX{\SetB \cup \brc{\query}}$. So, let $\CC = \CHX{\SetB \cup
   \brc{\query}}$, and $\VV$ be the set of vertices of $\CC$. Then, we
have that $\mprX{\query} \:=\: 1 - \Prob{\query \in \VV}$.

If $\SetB = \emptyset$, then clearly $\CC = \brc{\query}$ and $\query
\in \VV$. Otherwise, $|\VV| \ge 2$ and $\query \in \VV$ implies that
$\query$ is an endpoint of exactly two edges on the boundary of
$\CC$.\footnote{If $\SetB$ consists of a single site $\siX{i}$, then
   $\CC$ is the line segment $\seg{\query\siX{i}}$. In this case, we
   consider the boundary of $\CC$ to be a cycle formed by two edges:
   one going from $\query$ to $\siX{i}$, and one going from $\siX{i}$
   back to $\query$.}
\hakan{Changed this sentence.} %
In this case, the first edge following
$\query$ in the counter-clockwise order of $\CC$ is called the
\emphi{witness edge} of $\query$ being in $\VV$.  Hence, $\query \in
\VV$ if and
only if $\SetB = \emptyset$ or (exclusively) $\SetB$ has a witness
edge, i.e.,
\begin{align*}
    \Prob{\Bigl. \query \in \VV}%
    =%
    \Prob{\Bigl. \SetB = \emptyset} + \sum_{i=1}^n \Prob{\Bigl. \query
       \siX{i} \text{ is the witness edge of } \query \notin
       \CHX{\SetB}}.
\end{align*}
The first term can be computed in linear time. To compute the $i$\th
term in the summation, we observe that $\query\siX{i}$ is the witness
edge of $\SetB$ if and only if $\siX{i} \in \SetB$ and $\SetB$
contains no sites to the right of the oriented line spanned by the
vector $\vect{\query\siX{i}}$, and the corresponding probability is
\begin{math}
    \Bigl. \priX{i} \cdot \prod_{\siX{j} \in \GroupX{i}}
    \inv{\priX{j}},
\end{math}
where $\GroupX{i}$ is the set of sites to the right of
$\vect{\query\siX{i}}$. This expression can be computed in $O(n)$
time. It follows that one can compute $1 - \mprX{\query}$, and
therefore $\mprX{\query}$, in $O(n^2)$ time.
%
\hakan{Changed the next few sentences.} The computation time can be improved to $O(n\log{n})$
as described in the following paragraph.

\mparagraph{Improving the running time.} %
%
%
\begin{figure}[t]
    \begin{tabular}{cccc}
      \begin{minipage}{0.08\linewidth}
        ~
      \end{minipage}
      \begin{minipage}{0.44\linewidth}
        \centerline{
\definecolor{000000}{RGB}{0,0,0}
\definecolor{808080}{RGB}{128,128,128}
\definecolor{a0a0a4}{RGB}{160,160,164}
\begin{tikzpicture}[x=1cm,y=1cm]
\clip (6,0) rectangle (8,3);
\node [color=000000,inner sep=0.5pt,above,shift={(0,3pt)}] at
(6.99931,2.38207) {$\siX{2}$};
\node [color=000000,inner sep=0.5pt,above,shift={(0,3pt)}] at
(7.53016,1.82715) {$\siX{1}$};
\node [color=000000,inner sep=0.5pt,above right,shift={(3pt,3pt)}] at (6.0587,2.24139) {$...$};
\node [color=000000,inner sep=0.5pt,right,shift={(3pt,0)}] at (11.4412,1.66687) {$\mpt$};
\node [color=000000,inner sep=0.5pt,left,shift={(-3pt,0)}] at (6.47966,1.38936) {$\mpt$};
\node [color=000000,inner sep=0.5pt,below right,shift={(3pt,-3pt)}] at (6.37052,0.901597) {$...$};
\node [color=000000,inner sep=0.5pt,above,shift={(0,3pt)}] at (11.4617,0.338693) {$\circled{$W_i$}$};
\node [color=000000,inner sep=0.5pt,above,shift={(0,3pt)}] at
(7.40633,0.803769) {$\siX{n}$};
\node [color=000000,inner sep=0.5pt,right,shift={(3pt,0)}] at
(11.9193,2.41117) {$\siX{i-1}$};
\node [color=000000,inner sep=0.5pt,left,shift={(-3pt,0)}] at (10.8907,2.36988) {$\siX{i}$};
\node [color=000000,inner sep=0.5pt,right,shift={(3pt,0)}] at (2.14692,1.31496) {$\mpt$};
\draw[color=000000,line width=1] (1.6267,0.406649) -- (2.14692,1.31496);
\draw[color=000000,line width=1] (0.635818,0.637855) -- (1.6267,0.406649);
\draw[color=808080,line width=1,densely dashed] (6.47966,1.38936) -- (6.37052,0.901597);
\draw[color=a0a0a4,line width=1,dotted] (10.4354,0.101089) -- (12.4139,3.18114);
\draw[color=808080,line width=1,densely dashed] (6.47966,1.38936) -- (6.0587,2.24139);
\draw[color=808080,line width=1,densely dashed] (6.47966,1.38936) -- (7.40633,0.803769);
\draw[color=a0a0a4,line width=1,dotted] (10.3102,3.11133) -- (12.5682,0.227505);
\draw[color=808080,line width=1,densely dashed] (6.99931,2.38207) -- (6.47966,1.38936);
\draw[color=808080,line width=1,densely dashed] (7.53016,1.82715) -- (6.47966,1.38936);
\draw[color=000000,line width=1] (0.974371,2.54531) -- (0.404612,1.41405);
\draw[->,color=000000,line width=2] (2.14692,1.31496) -- (1.8414,2.18199);
\draw[color=000000,line width=1] (1.8414,2.18199) -- (0.974371,2.54531);
\draw[color=000000,line width=1] (0.404612,1.41405) -- (0.635818,0.637855);
\path [fill=000000] (6.99931,2.38207) circle (1.5pt);
\path [fill=000000] (6.47966,1.38936) circle (1.5pt);
\path [fill=000000] (7.53016,1.82715) circle (1.5pt);
\path [fill=000000] (7.40633,0.803769) circle (1.5pt);
\path [fill=000000] (10.8907,2.36988) circle (1.5pt);
\path [fill=000000] (11.4412,1.66687) circle (1.5pt);
\path [fill=000000] (11.9193,2.41117) circle (1.5pt);
\path [fill=000000] (0.974371,2.54531) circle (1.5pt);
\path [fill=000000] (1.6267,0.406649) circle (1.5pt);
\path [fill=000000] (0.635818,0.637855) circle (1.5pt);
\path [fill=000000] (0.404612,1.41405) circle (1.5pt);
\path [fill=000000] (1.8414,2.18199) circle (1.5pt);
\path [fill=000000] (2.14692,1.31496) circle (1.5pt);
\end{tikzpicture}
}
        \caption{Sites in radial order around $\query$.}
        \figlab{radial2}
      \end{minipage}
      &
      \begin{minipage}{0.44\linewidth}
        \centerline{
\definecolor{000000}{RGB}{0,0,0}
\definecolor{808080}{RGB}{128,128,128}
\definecolor{a0a0a4}{RGB}{160,160,164}
\begin{tikzpicture}[x=1cm,y=1cm]
\clip (10.3,0) rectangle (13,3);
\node [color=000000,inner sep=0.5pt,above,shift={(0,3pt)}] at (6.99931,2.38207) {$\siX{2}$};
\node [color=000000,inner sep=0.5pt,above,shift={(0,3pt)}] at (7.53016,1.82715) {$\siX{1}$};
\node [color=000000,inner sep=0.5pt,above right,shift={(3pt,3pt)}] at (6.0587,2.24139) {$...$};
\node [color=000000,inner sep=0.5pt,right,shift={(3pt,0)}] at (11.4412,1.66687) {$\mpt$};
\node [color=000000,inner sep=0.5pt,left,shift={(-3pt,0)}] at (6.47966,1.38936) {$\mpt$};
\node [color=000000,inner sep=0.5pt,below right,shift={(3pt,-3pt)}] at (6.37052,0.901597) {$...$};
\node [color=000000,inner sep=0.5pt,above,shift={(0,3pt)}] at (11.4617,0.338693) {$\circled{$W_i$}$};
\node [color=000000,inner sep=0.5pt,above,shift={(0,3pt)}] at
(7.40633,0.803769) {$\siX{n}$};
\node [color=000000,inner sep=0.5pt,right,shift={(3pt,0)}] at
(11.9193,2.41117) {$\siX{i-1}$};
\node [color=000000,inner sep=0.5pt,left,shift={(-3pt,0)}] at (10.8907,2.36988) {$\siX{i}$};
\node [color=000000,inner sep=0.5pt,right,shift={(3pt,0)}] at (2.14692,1.31496) {$\mpt$};
\draw[color=000000,line width=1] (1.6267,0.406649) -- (2.14692,1.31496);
\draw[color=000000,line width=1] (0.635818,0.637855) -- (1.6267,0.406649);
\draw[color=808080,line width=1,densely dashed] (6.47966,1.38936) -- (6.37052,0.901597);
\draw[color=a0a0a4,line width=1,dotted] (10.4354,0.101089) -- (12.4139,3.18114);
\draw[color=808080,line width=1,densely dashed] (6.47966,1.38936) -- (6.0587,2.24139);
\draw[color=808080,line width=1,densely dashed] (6.47966,1.38936) -- (7.40633,0.803769);
\draw[color=a0a0a4,line width=1,dotted] (10.3102,3.11133) -- (12.5682,0.227505);
\draw[color=808080,line width=1,densely dashed] (6.99931,2.38207) -- (6.47966,1.38936);
\draw[color=808080,line width=1,densely dashed] (7.53016,1.82715) -- (6.47966,1.38936);
\draw[color=000000,line width=1] (0.974371,2.54531) -- (0.404612,1.41405);
\draw[->,color=000000,line width=2] (2.14692,1.31496) -- (1.8414,2.18199);
\draw[color=000000,line width=1] (1.8414,2.18199) -- (0.974371,2.54531);
\draw[color=000000,line width=1] (0.404612,1.41405) -- (0.635818,0.637855);
\path [fill=000000] (6.99931,2.38207) circle (1.5pt);
\path [fill=000000] (6.47966,1.38936) circle (1.5pt);
\path [fill=000000] (7.53016,1.82715) circle (1.5pt);
\path [fill=000000] (7.40633,0.803769) circle (1.5pt);
\path [fill=000000] (10.8907,2.36988) circle (1.5pt);
\path [fill=000000] (11.4412,1.66687) circle (1.5pt);
\path [fill=000000] (11.9193,2.41117) circle (1.5pt);
\path [fill=000000] (0.974371,2.54531) circle (1.5pt);
\path [fill=000000] (1.6267,0.406649) circle (1.5pt);
\path [fill=000000] (0.635818,0.637855) circle (1.5pt);
\path [fill=000000] (0.404612,1.41405) circle (1.5pt);
\path [fill=000000] (1.8414,2.18199) circle (1.5pt);
\path [fill=000000] (2.14692,1.31496) circle (1.5pt);
\end{tikzpicture}
}
        \caption{The set $W_i$.}
        \figlab{switch2}
      \end{minipage}
      \begin{minipage}{0.04\linewidth}
        ~
      \end{minipage}
    \end{tabular}
\end{figure}
%
%
The main
idea is to compute the witness edge probabilities in radial order
around $\query$. We sort all sites in counter-clockwise order around
$\query$. Without loss of generality, assume that the circular
sequence $\siX{1},\ldots,\siX{n}$ is the resulting order. (See
\figref{radial2}.) We first compute the probability that
$\query\siX{1}$ is the witness edge in $O(n)$ time. Then, for
increasing values of $i$ from $2$ to $n$, we compute the probability
that $\query\siX{i}$ is the witness edge by updating the probability
for $\query\siX{i-1}$, in $O(1)$ amortized time. In particular, let
$W_i$ denote the set of sites in the open wedge bounded by the vectors
$\vect{\query\siX{i-1}}$ and $\vect{\query\siX{i}}$. (See
\figref{switch2}.)  Notice that $\GroupX{i} = \GroupX{i-1} \cup
\brc{\siX{i-1}} \setminus W_i$. It follows that the probability for
$\query\siX{i}$ can be computed by multiplying the probability for
$\query\siX{i-1}$ with $\frac{\priX{i}}{\priX{i-1}} \times
\frac{\inv{\priX{i-1}}}{\prod_{\siX{j} \in W_i}\inv{\priX{j}}}~.$ The
cost of a single update is $O(1)$ amortized because total number
multiplications in all the updates is at most $4n$. (Each site affects
at most $4$ updates.) Finally, notice that we can easily keep track of
the set $W_i$ during our radial sweep, as changes to this set follow
the same radial order. \hakan{Removed "The result"}


\begin{theorem}
    Given a set of $n$ uncertain points in $\Re^2$ under the
    {\unipoint} model, the membership probability of a query point
    $\query$ can be computed in $O(n\log{n})$ time.
\end{theorem}

\subsection{The $d$-dimensional case}

The difficulty in extending the above to higher dimensions is an
appropriate generalization of {witness edges}, which allow us to
implicitly sum over exponentially many outcomes without {overcounting}. Our algorithm requires that all sites, including the
\hakan{Changed this sentence (for now). See my email.} query point $\query$, are in general position, i.e., no $k+1$ points of
$\UPntSet \cup \brc{\query}$ lie on a $(k-1)$-hyperplane when projected into a subset of $k$ coordinates, where $2 \le k \le d$.

Let $\SetB$ be an outcome, $\CC = \CHX{ \SetB \cup \brc{\query}}$ its
convex hull, and $\VV$ the vertices of $\CC$.  Let $\lowd(\SetBA)$
denote the point with the lowest $x_d$-coordinate in $\SetBA$.
Clearly, if $\query$ is $\lowd(\SetBA)$ then $\query \in \VV$;
otherwise, we condition the probability based on which point among
$\SetB$ is $\lowd(\SetBA)$.  Therefore, we can write
\begin{align*}
    \Prob{\Bigl. \query \in \VV}%
    =%
    \Prob{\Bigl. \query \altequiv \lowd(\SetBA)} + \sum_{1 \le i \le
       n}{\Prob{\Bigl. \siX{i} \altequiv \lowd(\SetBA) \;\land\; \query \in
          \VV}}.
\end{align*}
It is easy to compute the first term. We show below how to compute
each term of the summation in $O(n^{d-1})$ time, which gives the
desired bound of $O(n^d)$.

\parpic[r]{
\definecolor{c4c4c9}{RGB}{196,196,201}
\definecolor{a0a0a4}{RGB}{160,160,164}
\definecolor{000000}{RGB}{0,0,0}
\definecolor{c3c3c3}{RGB}{195,195,195}
\begin{tikzpicture}[x=1cm,y=1cm,scale=0.79,every node/.style={scale=0.79}]
\clip (1,0.5) rectangle (6,4.2);
\path [fill,color=c4c4c9,line width=1] (4.68037,2.14606)  --  (3.05903,1.7406)  --  (4.26829,1.08977)  --  cycle;
\node [color=000000,inner sep=0.5pt,above,shift={(0,3pt)}] at (7.05561,2.99741) {$\CC$};
\node [color=000000,inner sep=0.5pt,below,shift={(0,-3pt)}] at (2.53464,1.90118) {$\siX{i}'$};
\node [color=000000,inner sep=0.5pt,below,shift={(0,-3pt)}] at (7.05561,1.43758) {$\CCP$};
\node [color=000000,inner sep=0.5pt,above,shift={(0,3pt)}] at (2.05561,2.99741) {$\CC$};
\node [color=000000,inner sep=0.5pt,below,shift={(0,-3pt)}] at (3.96115,3.05154) {$\f$};
\node [color=000000,inner sep=0.5pt,below,shift={(0,-3pt)}] at (2.05561,1.43758) {$\CCP$};
\path [fill,color=c3c3c3,line width=1] (4.68037,3.30326)  --  (3.05903,2.92653)  --  (4.26829,2.98349)  --  cycle;
\node [color=000000,inner sep=0.5pt,below,shift={(0,-3pt)}] at (3.05903,1.7406) {$\mpt'$};
\node [color=000000,inner sep=0.5pt,below right,shift={(3pt,-3pt)}] at (3.05903,2.92653) {$\mpt$};
\node [color=000000,inner sep=0.5pt,below,shift={(0,-3pt)}] at
(7.53464,1.90118) {$\siX{i}'$};
\node [color=000000,inner sep=0.5pt,above left,shift={(-3pt,3pt)}] at (9.26829,1.08977) {$\mpt'$};
\node [color=000000,inner sep=0.5pt,below,shift={(0,-3pt)}] at (12.839,1.53555) {$\mpt'$};
\node [color=000000,inner sep=0.5pt,above right,shift={(3pt,3pt)}] at (13.8386,1.10418) {$$};
\node [color=000000,inner sep=0.5pt,below,shift={(0,-3pt)}] at (11.9983,1.89834) {$\siX{i}'$};
\node [color=000000,inner sep=0.5pt,above,shift={(0,3pt)}] at (11.9983,3.58425) {$\siX{i}$};
\node [color=000000,inner sep=0.5pt,above,shift={(0,3pt)}] at (12.839,3.05073) {$\mpt$};
\draw[color=a0a0a4,line width=1] (3.05903,2.92653) -- (4.68037,3.30326);
\draw[color=000000,line width=1] (11.1791,2.37338) -- (11.748,0.676271);
\draw[color=000000,line width=1] (1.74802,0.676271) -- (5.54352,0.676271);
\draw[color=000000,line width=1] (1.17914,2.37338) -- (1.74802,0.676271);
\draw[color=000000,line width=1] (11.1791,2.37338) -- (14.8199,2.37338);
\draw[color=000000,line width=1] (11.748,0.676271) -- (15.5435,0.676271);
\draw[color=000000,line width=1] (15.5435,0.676271) -- (14.8199,2.37338);
\draw[color=000000,line width=1] (5.54352,0.676271) -- (4.81987,2.37338);
\draw[color=000000,line width=1] (6.74802,0.676271) -- (10.5435,0.676271);
\draw[color=000000,line width=1] (10.5435,0.676271) -- (9.81987,2.37338);
\draw[color=000000,line width=1] (1.17914,2.37338) -- (4.81987,2.37338);
\draw[color=a0a0a4,line width=1] (8.05903,2.92653) -- (9.68037,3.30326);
\draw[color=000000,line width=1] (6.17914,2.37338) -- (6.74802,0.676271);
\draw[color=a0a0a4,line width=1] (1.80624,4.08849) -- (4.68037,3.30326);
\draw[color=a0a0a4,line width=1] (6.80624,4.08849) -- (9.68037,3.30326);
\draw[color=000000,line width=1] (6.17914,2.37338) -- (9.81987,2.37338);
\draw[color=a0a0a4,line width=1,dotted] (9.26829,2.98349) -- (9.26829,1.08977);
\draw[color=000000,line width=1] (2.78649,3.43984) -- (4.26829,2.98349);
\draw[color=a0a0a4,line width=1,dotted] (11.9983,1.89834) -- (12.839,1.53555);
\draw[color=000000,line width=1] (7.78649,0.916823) -- (9.26829,1.08977);
\draw[->,color=000000,line width=1] (1,-1) -- (6,-1);
\draw[color=000000,line width=1] (9.26829,2.98349) -- (8.05903,2.92653);
\draw[->,color=000000,line width=1] (3.05903,1.7406) -- (3.99269,1.45469);
\draw[color=000000,line width=1] (6.80624,4.08849) -- (9.79896,3.9566);
\draw[color=000000,line width=1] (7.78649,3.43984) -- (9.26829,2.98349);
\draw[color=a0a0a4,line width=1,dotted] (2.53464,1.90118) -- (3.05903,1.7406);
\draw[color=000000,line width=1] (6.80624,2.18721) -- (7.78649,0.916823);
\draw[color=000000,line width=1] (4.26829,2.98349) -- (4.79896,3.9566);
\draw[->,color=000000,line width=1] (-1.1336,-1.90284) -- (-1.1336,-0.902838);
\draw[color=000000,line width=1] (1.80624,4.08849) -- (4.79896,3.9566);
\draw[color=000000,line width=1] (4.26829,2.98349) -- (4.68037,3.30326);
\draw[color=000000,line width=1] (2.78649,3.43984) -- (3.05903,2.92653);
\draw[color=000000,line width=1] (1.80624,2.18721) -- (2.78649,0.916823);
\draw[color=000000,line width=1] (2.78649,3.43984) -- (4.79896,3.9566);
\draw[color=000000,line width=1] (6.80624,4.08849) -- (8.05903,2.92653);
\draw[->,color=000000,line width=1] (-1.1336,-1.90284) -- (-2.1336,-1.90284);
\draw[color=000000,line width=1] (9.26829,2.98349) -- (9.79896,3.9566);
\draw[color=000000,line width=1] (9.68037,3.30326) -- (9.79896,3.9566);
\draw[color=000000,line width=1] (4.68037,3.30326) -- (4.79896,3.9566);
\draw[color=000000,line width=1] (4.26829,2.98349) -- (3.05903,2.92653);
\draw[color=000000,line width=1] (1.80624,4.08849) -- (3.05903,2.92653);
\draw[color=000000,line width=1] (2.78649,0.916823) -- (4.26829,1.08977);
\draw[color=000000,line width=1] (4.79896,1.65271) -- (4.68037,2.14606);
\draw[color=000000,line width=1] (7.78649,3.43984) -- (9.79896,3.9566);
\draw[color=000000,line width=1] (7.78649,3.43984) -- (8.05903,2.92653);
\draw[color=a0a0a4,line width=1,dotted] (7.53464,1.90118) -- (9.26829,1.08977);
\draw[->,color=000000,line width=1] (9.26829,1.08977) -- (10.1114,0.69518);
\draw[color=000000,line width=1] (9.26829,2.98349) -- (9.68037,3.30326);
\draw[color=000000,line width=1] (9.79896,1.65271) -- (9.68037,2.14606);
\draw[->,color=000000,line width=1] (12.839,1.53555) -- (13.8386,1.10418);
\draw[color=a0a0a4,line width=1,dotted] (12.839,3.05073) -- (12.839,1.53555);
\draw[color=a0a0a4,line width=1,dotted] (11.9983,3.58425) -- (11.9983,1.89834);
\draw[color=000000,line width=1] (4.26829,1.08977) -- (4.79896,1.65271);
\draw[color=000000,line width=1] (1.80624,2.18721) -- (4.68037,2.14606);
\draw[color=a0a0a4,line width=1,dotted] (3.05903,2.92653) -- (3.05903,1.7406);
\draw[color=a0a0a4,line width=1,dotted] (7.53464,3.76949) -- (7.53464,1.90118);
\draw[color=a0a0a4,line width=1,dotted] (2.53464,3.76949) -- (2.53464,1.90118);
\path [fill=000000] (9.79896,3.9566) circle (1.5pt);
\path [fill=000000] (3.05903,1.7406) circle (1.5pt);
\path [fill=000000] (9.26829,1.08977) circle (1.5pt);
\path [fill=000000] (4.68037,3.30326) circle (1.5pt);
\path [fill=000000] (12.839,3.05073) circle (1.5pt);
\path [fill=000000] (11.9983,3.58425) circle (1.5pt);
\path [fill=000000] (12.839,1.53555) circle (1.5pt);
\path [fill=000000] (11.9983,1.89834) circle (1.5pt);
\path [fill=000000] (8.05903,2.92653) circle (1.5pt);
\path [fill=000000] (3.05903,2.92653) circle (1.5pt);
\path [fill=000000] (1.80624,4.08849) circle (1.5pt);
\path [fill=000000] (4.79896,3.9566) circle (1.5pt);
\path [fill=000000] (4.26829,1.08977) circle (1.5pt);
\path [fill=000000] (2.53464,3.76949) circle (1.5pt);
\path [fill=000000] (2.53464,1.90118) circle (1.5pt);
\path [fill=000000] (9.26829,2.98349) circle (1.5pt);
\path [fill=000000] (4.79896,1.65271) circle (1.5pt);
\path [fill=000000] (6.80624,2.18721) circle (1.5pt);
\path [fill=000000] (2.78649,0.916823) circle (1.5pt);
\path [fill=000000] (2.78649,3.43984) circle (1.5pt);
\path [fill=000000] (1.80624,2.18721) circle (1.5pt);
\path [fill=000000] (7.78649,3.43984) circle (1.5pt);
\draw[color=000000,line width=1] (1.80624,4.08849) -- (2.78649,3.43984);
\path [fill=000000] (1,-1) circle (1.5pt);
\path [fill=000000] (4.68037,2.14606) circle (1.5pt);
\path [fill=000000] (9.79896,1.65271) circle (1.5pt);
\draw[color=000000,line width=1] (6.80624,4.08849) -- (7.78649,3.43984);
\path [fill=000000] (6,-1) circle (1.5pt);
\path [fill=000000] (9.68037,3.30326) circle (1.5pt);
\path [fill=000000] (7.78649,0.916823) circle (1.5pt);
\path [fill=000000] (9.68037,2.14606) circle (1.5pt);
\path [fill=000000] (7.53464,3.76949) circle (1.5pt);
\path [fill=000000] (4.26829,2.98349) circle (1.5pt);
\path [fill=000000] (6.80624,4.08849) circle (1.5pt);
\path [fill=000000] (7.53464,1.90118) circle (1.5pt);
\node [color=000000,inner sep=0.5pt,above,shift={(0,5pt)},scale=0.81] at (3.86467,1.43312) {$\rray$};
\node [color=000000,inner sep=0.5pt,below right,shift={(3pt,-3pt)}] at (9.26829,2.98349) {$\mpt$};
\node [color=000000,inner sep=0.5pt,right,shift={(3pt,0)}] at (7.53464,3.76949) {$\siX{i}$};
\node [color=000000,inner sep=0.5pt,right,shift={(3pt,0)}] at (2.53464,3.76949) {$\siX{i}$};
\end{tikzpicture}
}
%

Consider an outcome $\SetB$ with $\siX{i} \in \SetB$. Let $\SetB',
\siX{i}'$ and $\query'$ denote the projections of $\SetB$, $\siX{i}$
and $\query$ respectively on the hyperplane $x_d = 0$, which we
identify with $\Re^{d-1}$. Let us define $\CCP = \CHX{\SetBAP} \subset
\Re^{d-1}$, and let $\VVP$ be the vertices of $\CCP$.

Let $\rray$ denote the open ray emanating from $\query'$ in the
direction of the vector $\vect{\siX{i}' \query'}$ (that is, this ray
is moving ``away'' from $\siX{i}'$).  A facet $\f$ of $\CC$ is a
\emphi{$\siX{i}$-escaping} facet for $\query$, if $\query$ is a vertex
of $\f$ and the projection of $\f$ on $\Re^{d-1}$ intersects
$\rray$. See the figure on the right. The following lemma is key to
our algorithm. The points of $\CC$ projected into $\partial \CCP$ form
the \emphi{silhouette} of $\CC$.



\remove{%
\begin{figure}[t]
    \centering 
\definecolor{000000}{RGB}{0,0,0}
\definecolor{585858}{RGB}{88,88,88}
\definecolor{c3c3c3}{RGB}{195,195,195}
\definecolor{585858}{RGB}{88,88,88}
\definecolor{808080}{RGB}{128,128,128}
\begin{tikzpicture}[x=1cm,y=1cm]
\clip (1,0) rectangle (6,4.7);
\path [fill,color=c3c3c3,line width=1] (5.31943,1.63805)  --  (1.80931,2.56302)  --  (1.80931,4.01943)  --  (5.31943,4.01943)  --  cycle;
\node [color=000000,inner sep=0.5pt,below,shift={(0,-3pt)}] at (3.28219,0.655841) {$\mpt'$};
\path [fill,color=c3c3c3,line width=1] (5.31943,2.97328)  --  (5.31943,0.655841)  --  (1.80931,0.655841)  --  (1.80931,1.59769)  --  cycle;
\node [color=000000,inner sep=0.5pt,below,shift={(0,-3pt)}] at (5.31943,0.655841) {$\rray$};
\node [color=000000,inner sep=0.5pt,above,shift={(0,3pt)}] at (1.80931,4.01943) {$d\mbox{th axis}$};
\path [fill,color=808080,line width=1] (5.31943,2.97328)  --  (3.28219,2.1749)  --  (5.31943,1.63805)  --  cycle;
\draw[->,color=000000,line width=1] (-1.1336,-1.90284) -- (-1.1336,-0.902838);
\draw[->,color=000000,line width=1] (-1.1336,-1.90284) -- (-2.1336,-1.90284);
\draw[->,color=000000,line width=1] (1,-1) -- (6,-1);
\draw[color=000000,line width=1,dotted] (5.31943,2.97328) -- (3.28219,2.1749);
\draw[color=000000,line width=1,dotted] (1.80931,1.59769) -- (3.28219,2.1749);
\draw[color=000000,line width=1,dotted] (1.80931,2.56302) -- (3.28219,2.1749);
\draw[->,color=000000,line width=1] (1.80931,0.655841) -- (5.31943,0.655841);
\draw[->,color=000000,line width=1] (1.80931,0.655841) -- (1.80931,4.01943);
\draw[color=000000,line width=1,dotted] (5.31943,1.63805) -- (3.28219,2.1749);
\draw[->,color=000000,line width=2] (3.28219,0.655841) -- (5.31943,0.655841);
\draw[color=585858,line width=3] (3.28219,2.1749) -- (4.56037,1.83808);
\draw[color=585858,line width=3] (4.54583,2.67011) -- (3.28219,2.1749);
\path [fill=000000] (1,-1) circle (1.5pt);
\path [fill=000000] (6,-1) circle (1.5pt);
\path [fill=000000] (3.28219,2.1749) circle (2pt);
\path [fill=000000] (3.28219,0.655841) circle (2pt);
\node [color=000000,inner sep=0.5pt,left,shift={(-6pt,0)}] at (5.31943,2.30566) {$\CC$};
\node [color=000000,inner sep=0.5pt,below,shift={(0,-3pt)}] at (3.28219,2.1749) {$\mpt$};
\node [color=000000,inner sep=0.5pt,below,shift={(0,-3pt)}] at (3.92128,2.00649) {$\f_1$};
\node [color=000000,inner sep=0.5pt,above,shift={(0,3pt)}] at (3.91401,2.4225) {$\f_2$};
\node [color=000000,inner sep=0.5pt,above,shift={(0,3pt)}] at (3.28219,0.867205) {$\HYP_2$};
\node [color=000000,inner sep=0.5pt,above,shift={(0,3pt)}] at (3.28219,3.06963) {$\HYP_1$};
\end{tikzpicture}

    \caption{The cross-section of the space on the plane defined by
       the $d$\th coordinate axis and the line supporting
       $\protect\rray$.}
    \figlab{plot}
\end{figure}%
}

\begin{lemma}
    \lemlab{projection}%
    \begin{inparaenum}[(A)]
        \item \hakan{Changed this a little to make probability decomposition more transparent. Please, check. Original was ``if $\query' \in \VVP$ then $\query$ is a (silhouette) vertex of $\CC$.} If $\query' \in \VVP$ then $\query$ is a silhouette vertex of $\CC$ and vice versa.
        
        \item If $\siX{i} \in \SetB$ then $\query$ has at most one
        $\siX{i}$-escaping facet on $\CC$.

        \item The point $\query$ is a non-silhouette vertex of the
        convex-hull $\CC$ if and only if $\query$ has a (single)
        $\siX{i}$-escaping facet on $\CC$.
    \end{inparaenum}
\end{lemma}

\begin{proof}
    (A) \hakan{This proof is vague to me. We need to improve it at some point if we want to go for a journal.} By definition.

    (B) If $\query$ has a $\siX{i}$-escaping facet then it is a vertex
    of the convex-hull $\CC$. Consider the union of facets adjacent to
    $\query$, and observe that the projection of this ``tent'' can
    fold over itself in the projection only if $\query$ is on the
    silhouette. Specifically, if $\query$ is not on the silhouette
    then the claim immediately holds. 

    Otherwise, $\query$ is on the silhouette then the open ray 
    $\rray$ does not intersect $\CCP$, and there are no 
    $\siX{i}$-escaping facets.

    (C) Follows immediately from (B), by observing that in this case,
    the projected ``tent'', surrounds $\query'$, and as such one of
    the facets must be an escaping facets for $\siX{i}$.
\end{proof}

Given a subset of sites $\Sites_{\alpha} \subseteq \Sites \setminus
\brc{\siX{i}}$ of size $(d-1)$, define $\f(\SFACE)$ to be the
$(d-1)$-dimensional simplex $\CHX{\Sites_{\alpha} \cup
   \brc{\query}}$. Since $\siX{i} \altequiv \lowd(\SetBA)$ implies
$\siX{i} \in \SetB$, we can use \lemref{projection} to decompose the
$i$\th term as follows:
\begin{align*}
    &%
    \Prob{\Bigl. \siX{i} \altequiv \lowd(\SetBA) \;\land\; \query \in \VV}%
    =%
    \Prob{\Bigl. \siX{i} \altequiv \lowd(\SetBA) \;\land\; \query' \in
       \VVP}%
    \\& \qquad%
    + \sum_{\substack{ \SFACE \subseteq {\Sites \setminus \brc{\siX{i}}} \\
          |\SFACE|=(d-1) \\ \f(\SFACE)%
          \text{ is }\siX{i}\text{-escaping for $q$}}}%
    \Prob{\Bigl. \siX{i} \altequiv \lowd(\SetBA) \;\land\; \f(\SFACE)
       \text{ is a facet of } \CC}.
\end{align*}
The first term is an instance of the same problem in $(d-1)$
dimensions (for the point $\query'$ and the projection of $\Sites$),
and thus is computed recursively. For the second term, we compute the
probability that $\f(\SFACE)$ is a facet of $\CC$ as follows. Let
$\GroupX{1} \subseteq \Sites$ be the subset of sites which are on the
other side of the hyperplane supporting $\f(\SFACE)$ with respect to
$\siX{i}$. Let $\GroupX{2} \subseteq \Sites$ be the subset of sites
that are below $\siX{i}$ along the $x_d$-axis. Clearly, $\f(\SFACE)$
is a facet of $\CC$ (and $\siX{i} \altequiv \lowd(\SetBA)$) if and
only if all points in $\SFACE$ and $\siX{i}$ exist in $\SetB$, and all
points in $\GroupX{1} \cup \GroupX{2}$ are absent from $\SetB$. The
corresponding probability can be written as
\begin{flalign*}
    \priX{i} \times \prod_{\siX{j} \in \SFACE}{\priX{j}} \times
    \prod_{ \siX{j} \,\in\, \GroupX{1} \cup \GroupX{2}
    }{\inv{\priX{j}}} \,.
\end{flalign*}
This formula is valid only if $\SFACE \cap \GroupX{2} = \emptyset$ and
$\siX{i}$ has a lower $x_d$-coordinate than $\query$; otherwise we set
the probability to zero.  This expression can be computed in linear
time, and the whole summation term can be computed in $O(n^d)$
time. Then, by induction, the computation of the $i$\th term takes
$O(n^d)$ time. 
\hakan{Changed the next few sentences in an attempt to address reviewer comment. Feel free to roll back if you don't like.} Notice that the base case of our induction requires computing the probability $\Probsmall{\Bigl. \siX{i} \altequiv \lowd(\SetBA) \;\land\; \query^{(d-2)} \in \VV^{(d-2)}}$ (where $^{(d-2)}$ indicates a projection to $\Re^2$). Computing this probability is essentially a two-dimensional membership probability problem on $\query$ and $P$, but is conditioned on the existence of $\siX{i}$ and the non-existence of all sites below $\siX{i}$ along $d$th axis. Our two dimensional algorithm can be easily adapted to solve this variation in $O(n\log{n})$ time as well.
%
(Briefly, we apply the same algorithm but we ignore all points that are below $\siX{i}$ . We later adjust the 
Finally, we can improve the computation time for the $i$\th term to
$O(n^{d-1})$ by considering the facets $\f(\SFACE)$ in radial order.
\InNotESAVer{The details can be found in \apndref{radial:D:D}.}%
\SeeFullVer{}

\mparagraph{Remark.} %
The degeneracy of the input is easy to handle in two dimensions, but
creates some technical difficulties in higher dimensions that we are
currently investigating.

\begin{theorem}
    Let $\UPntSet$ be an uncertain set of $n$ points in the
    {\unipoint} model in $\Re^d$ and $\query$ be a point. If the input
    sites and $q$ are in general position, then one can compute the
    membership probability of $\query$ in $O(n^d)$ time, using linear
    space.
\end{theorem}


\mparagraph{Extension to the {\multipoint} model.} %
The algorithm extends to the {\multipoint} model easily by modifying
the computation of the probability for an edge or facet.
\InNotESAVer{Deferring the details to \apndref{memmulti}, we conclude
   the following.}  \SeeFullVer{}

\begin{theorem}
    Given an uncertain set $\UPntSet$ of $n$ points in the
    {\multipoint} model in $\Re^d$ and a point $\query\in \Re^d$, we can
    compute the membership probability of $\query$ in $O(n\log{n})$ time
    for $d=2$, and in $O(n^d)$ time for $d \ge 3$ if input sites and
    $q$ are in general position.
\end{theorem}

%
%

\section{Membership Queries}
\seclab{mem_queries}

We describe two algorithms -- one deterministic and one Monte Carlo --
for preprocessing a set of uncertain points for efficient
membership-probability queries.

\mparagraph{Probability map.} %
The \emphi{probability map} $\mathbb{M}(\UPntSet)$ is the subdivision
of $\Re^d$ into maximal connected regions so that $\mprX{q}$ is the
same for all query points $q$ in a region. The following lemma gives a
tight bound on the size of $\mathbb{M}(\UPntSet)$.

\begin{lemma}
    \lemlab{pvd}%
    The worst-case complexity of the probability map of a set of
    uncertain points in $\Re^d$ is $\Theta(n^{d^2})$, under both the
    unipoint and the multipoint model, where $n$ is the total number
    of sites in the input.
\end{lemma}
\begin{proof}
    We prove the result for the unipoint model, as the extension to
    the multipoint is straightforward.  For the upper bound, consider
    the set $\HyperplaneSet$ of $O(n^d)$ hyperplanes formed by all
    $d$-tuples of points in $\UPntSet$. In the arrangement
    $\Arrangement(\HyperplaneSet)$ formed by these planes, each (open)
    cell has the same value of $\mprX{q}$. This arrangement, which is
    a \hakan{Is refinement the right word here?} refinement of $\mathbb{M}(\UPntSet)$, has size $O((n^d)^d) =
    O(n^{d^2})$, establishing the upper bound.
    
    \parpic[r]{\includegraphics[scale=0.34]%
       {\si{figs/r__8vert__p0_15__tries_2000}}}%
    
    For the lower bound, consider the problem in two dimensions;
    extension to higher dimensions is straightforward. We choose the
    sites to be the vertices $p_1, \ldots, p_n$ of a regular $n$-gon,
    where each site exists with probability $\w$, $0 < \w < 1$.  See
    the figure on the right.  Consider the arrangement $\Arrangement$
    formed by the line segments $p_ip_j$, $1 \leq i < j \leq n$, and
    treat each face as relatively open. If $\mprX{f}$ denotes the
    membership probability for a face $f$ of $\Arrangement$, then for
    any two faces $f_1$ and $f_2$ of $\Arrangement$, where $f_1$
    bounds $f_2$ (i.e., $f_{1} \subset \partial f_{2}$), we have $\mprX{f_1} \geq \mprX{f_2}$, and
    $\mprX{f_1} > \mprX{f_2}$ if $\w <1$.  Thus, the size of the
    arrangement $\Arrangement$ is also a lower bound on the complexity
    of $\mathbb{M}(\UPntSet)$. This proves that the worst-case
    complexity of $\mathbb{M}(\UPntSet)$ in $\Re^d$ is
    $\Theta(n^{d^2})$.
\end{proof}

We can preprocess this arrangement into a point-location data
structure, giving us the following result for $d=2$.

\begin{theorem}
    Let $\UPntSet$ be a set of uncertain points in $\Re^2$, with a
    total of $n$ sites. $\UPntSet$ can be preprocessed in $O(n^4)$
    time into a data structure of size $O(n^4)$ so that for any point
    $q \in \Re^d$, $\mprX{q}$ can be computed in $O(\log n)$ time.
\end{theorem}

\SeeFullVer{}%
\InNotESAVer{%
   \apndref{improve} describes how to construct the data structure in
   $O(n^4)$ time.%
}

\mparagraph{Remark.} %
\hakan{See my e-mail. We may need to mention perturbation.} For $d \ge 3$, due to our general position assumption, we can compute
the membership probability only for $d$-faces of
$\mathbb{M}(\UPntSet)$, and not for the lower-dimensional faces.  In
that case, by utilizing a point-location technique in
\cite{c-chfda-93}, one can build a structure that can report the
membership probability of a query point (inside a $d$-face) in
$O(\log{n})$ time, with a preprocessing cost of $O(n^{d^2+d})$.

\mparagraph{Monte Carlo algorithm.} %
The size of the probability map may be prohibitive even for $d = 2$,
so we describe a simple, space-efficient Monte Carlo approach for
quickly approximating the membership probability, within absolute
error. Fix a parameter $s > 1$, to be specified later. The
preprocessing consists of $s$ rounds, where the algorithm creates an
outcome $A_j$ of $\UPntSet$ in each round $j$. Each $A_j$ is
preprocessed into a data structure so that for a query point $q \in
\Re^d$, we can determine whether $q \in \CHX{A_j}$.

For $d \leq 3$, we can build each $\CHX{A_j}$ explicitly and use
linear-size point-location structures with $O(\log n)$ query
time. This leads to total preprocessing time $O(sn\log n)$ and space
$O(sn)$.  For $d \geq 4$, We use the data structure
in~\cite{ms-loq-92} for determining whether $q \in A_j$, for all $1
\leq j \leq s$. For a parameter $\tradeoffPara$ such that $n \leq
\tradeoffPara \leq n^{\lfloor d/2\rfloor}$ and for any constant
$\fixedvalue > 0$, using $O(s\tradeoffPara^{1+\fixedvalue})$ space and
preprocessing, it can compute in $O(\frac{sn}{\tradeoffPara^{1/\lfloor
      d/2\rfloor}} \log^{2d + 1} n)$ time whether $q \in \CHX{A_j}$
for every $j$.

Given a query point $q\in\Re^d$, we check for membership in all
$\CHX{A_j}$, and if it lies in $k$ of them, we return $\hmprX{q} =
k/s$ as our estimate of $\mprX{q}$.  Thus, the query time is
$O(\frac{sn}{\tradeoffPara^{1/\lfloor d/2\rfloor}} \log^{2d + 1} n)$
for $d \geq 4$, $O(s \log n)$ for $d = 3$, and $O(\log n + s)$ for $d
= 2$ (using fractional cascading).

It remains to determine the value of $s$ so that $|\mprX{q} -
\hmprX{q}| \leq \eps$ for all queries $q$, with probability at least
$1-\delta$.  For a fixed $q$ and outcome $A_j$, let $X_i$ be the
random indicator variable, which is 1 if $q \in \CHX{A_j}$ and 0
otherwise.  Since $\mathsf{E}[X_i] = \mprX{q}$ and $X_i \in
\brc{0,1}$, using a Chernoff-Hoeffding bound on $\hmprX{q} = k/s =
(1/s) \sum_i X_i$, we observe that $ \Prob{| \hmprX{q} - \mprX{q}|
   \geq \eps} \leq 2 \exp(-2\eps^2 s) \leq \delta'.  $ By
\lemref{pvd}, we need to consider $O(n^{d^2})$ distinct queries. If we
set $1/\delta' = O(n^{d^2}/ \delta)$ and $s = O((1/\eps^2)
\log(n/\delta))$, we obtain the following theorem.

\begin{theorem}
    Let $\UPntSet$ be a set of uncertain points in $\Re^d$ under the
    {\multipoint} model with a total of $n$ sites, and let
    $\eps,\delta \in (0, 1)$ be parameters. For $d \geq 4$, $\UPntSet$
    can be preprocessed, for any constant $\fixedvalue > 0$, in
    $O((\tradeoffPara^{1+\fixedvalue} /\eps^2) \log\frac{n}{\delta})$
    time, into a data structure of size
    $O((\tradeoffPara^{1+\fixedvalue} /\eps^2) \log\frac{n}{\delta})$,
    so that with probability at least $1-\delta$, for any query point
    $q \in \Re^2$, $\hmprX{q}$ satisfying $|\mprX{q} - \hmprX{q}| \leq
    \eps$ and $\hmprX{q} > 0$ can be returned in
    $O(\frac{n}{\tradeoffPara^{1/\lfloor
          d/2\rfloor}\eps^2}\log\frac{n}{\delta}\log^{2d + 1} n)$
    time, where $\tradeoffPara$ is a parameter and $n \leq
    \tradeoffPara \leq n^{\lfloor d/2\rfloor}$. For $d \leq 3$, the
    preprocessing time and space are
    $O(\frac{n}{\eps^2}\log\log\frac{n}{\delta}\log n)$ and
    $O(\frac{n}{\eps^2}\log\frac{n}{\delta})$, respectively. The query
    time is $O(\frac{1}{\eps^2}\log(\frac{n}{\delta}) \log n)$
    (resp. $O(\frac{1}{\eps^2}\log\frac{n}{\delta})$) for $d = 3$
    (resp. $d = 2$).
\end{theorem}

%
%

\section{Tukey Depth and Convex Hull}
\seclab{tukey}
The membership probability is neither a convex nor a continuous function,
as suggested by the example in the proof of \lemref{pvd}. In this
section, we establish a helpful structural property of this function,
intuitively showing that the probability stabilizes once we go deep
enough into the ``region''. Specifically, we show a connection between
the Tukey depth of a point $q$ with its membership probability; in two
dimensions, this also results in an efficient data structure for
approximating $\probX{q}$ quickly within a small absolute error.

\mparagraph{Estimating $\probX{q}$.} %
Let $Q$ be a set of weighted points in $\Re^d$. For a subset $A
\subseteq Q$, let $w(A)$ be the total weight of points in $A$. Then
the \emphi{Tukey depth} of a point $q \in \Re^d$ with respect to $Q$,
denoted by $\tukey(q, Q)$, is $\min w(Q\cap H)$ where the minimum is
taken over all halfspaces $H$ that contain $q$.%
\footnote{%
   If the points in $Q$ are unweighted, then $\tukey(q, Q)$ is simply
   the minimum number of points that lie in a closed halfspace that
   contains $q$.}  If $Q$ is obvious from the context, we use
$\tukey(q)$ to denote $\tukey(q, Q)$. Before bounding $\probX{q}$ in
terms of $\tukey(q, Q)$, we prove the following lemma.

\begin{lemma}%
    \lemlab{decomposition}%
    Let $Q$ be a finite set of points in $\Re^d$. For any $\pnt \in
    \Re^d$, there is a set $\SimplexSet = \brc{\simplex_1, \ldots,
       \simplex_T}$ of $d$-simplices formed by $Q$ such that
    \begin{inparaenum}[(i)]
        \item each $\simplex_i$ contains $\pnt$ in its interior;
        \item no pair of them shares a vertex; and
        \item $T \geq \ceil{\tukey(\pnt, Q) /d }$.
    \end{inparaenum}
\end{lemma}

\begin{proof}
    As long as $\tukey(\pnt, Q) > 0$, $\pnt\in \CHX{Q}$, and by
    \Caratheodory Theorem \cite{j-hrctt-93}, there is a $d$-simplex
    $\simplex$ with its $d+1$ vertices in $Q$ such that $\pnt \in
    \simplex$.  Remove the vertices of $\simplex$ from $Q$, and repeat
    the argument. Let $\simplex_1, \ldots, \simplex_T$ be the
    resulting simplices. Observe that at most $d$ vertices of
    $\simplex$ can be in an halfspace passing through $\pnt$, which
    implies that the Tukey depth of $\pnt$ drops by at most $d$ after
    each iteration of this algorithm. Hence $T \geq \ceil{\tukey(\pnt, Q) /
       d}$.
\end{proof}

We now use \lemref{decomposition} to bound $\probX{\pnt}$ in terms of
$\tukey(\pnt, P)$.

\begin{theorem}
    \theolab{eps:net}%
    Let $\UPntSet$ be a set of $n$ uncertain points in the uniform
    {\unipoint} model, that is, each point is chosen with the same
    probability $\pr > 0$. Let $P$ be the set of sites in
    $\UPntSet$. There is a constant $c > 0$ such that for any point
    $\pnt \in \Re^d$ with $\tukey(\pnt, P) = t$, we have
    \begin{math}
        (1-\pr)^t%
        \leq%
        1 - \probX{\pnt}%
        \leq%
        d \exp\Bigl(-\frac{\pr t}{cd^2}\Bigr).
    \end{math}
\end{theorem}

\begin{proof}
    For the first inequality, fix a closed halfspace $H$ that contains
    $t$ points of $P$. If none of these $t$ points is chosen then
    $\pnt$ does not appear in the convex hull of the outcome, so $1 -
    \probX{\pnt} \geq (1-\pr)^t$.
    
    Next, let $\SimplexSet$ be the set of simplices of
    \lemref{decomposition}, and let $V$ be its set of vertices, where
    $T \geq \ceil{\tDepth /d }$. Let $n' = \cardin{V} = (d+1)T$.  Set
    $\eps = \frac{1}{d+1}$. A random subset of $V$ of size
    $O(\frac{d}{\eps}\log\frac{1}{\eps\delta}) = O(d^2
    \log\frac{d}{\delta})$ is an $\eps$-net for halfspaces, with
    probability at least $1-\delta$.
    
    In particular, any halfspace passing through $\pnt$, contains at
    least $T$ points of $V$. That is, all these halfspaces are
    $\eps$-heavy and would be stabbed by an $\eps$-net. Now, if we
    pick each point of $V$ with probability $\pr$, it is not hard to
    argue that the resulting sample $\RSample$ is an
    $\eps$-net\footnote{The standard argument uses slightly different
       sampling, but this is a minor technicality, and it is not hard to
       prove the $\eps$-net theorem with this modified sampling
       model.}. Indeed, the expected size (and in with sufficiently
    large probability) of $\RSample \cap V$ is $ n''%
    =%
    n' \pr%
    =%
    (d+1)T \pr %
    \geq%
    t \pr.  $ As such, for some constant $c$, we need the minimal
    value of $\delta$ such that the inequality $t \pr \geq c d^2
    \ln\frac{d}{\delta}$ holds, which is equivalent to
    $\exp\pth{\frac{t \pr}{cd^2}} \geq \frac{d}{\delta}$. This in turn
    is equivalent to
    \begin{math}
        \delta \geq {d} \exp\pth{- \frac{t \pr}{cd^2}}.
    \end{math}
    Thus, we set $\delta = {d} \exp\pth{- \frac{t \pr}{cd^2}}$.
    
    Now, with probability at least $1-\delta$, for a point $\pnt$ in
    $\Re^d$, with Tukey depth at least $t$, we have that $\pnt$ is in
    the convex-hull of the sample.
\end{proof}

\mparagraph{Remark.} %
\theoref{eps:net} can be extended to the {\multipoint} model. Assuming
that each uncertain point has $n_i$ sites and each site is chosen with
probability $\pr$, one can show that
\begin{math}
    (1-\pr)^t \leq 1 - \probX{\pnt}%
    \leq%
    d \exp\pth{-\frac{\pr t}{cd^2 n^*}},
\end{math}
where $n^* = \max_{1 \le i \le m} n_i$.

\theoref{eps:net} can be extended to the case when each point $p_i$ of
$\UPntSet$ is chosen with different probability, say, $\pr_i$.  In
order to apply \theoref{eps:net}, we convert $\UPntSet$ to a multiset
$\PntSetA$, as follows. We choose a parameter
$\eta=\tfrac{\delta}{10n}$. For each point $p_i \in \UPntSet$, we make
$w_i = \ceil{ \displaystyle\frac{\ln (1-\pr_i)}{\ln (1-\eta)}}$ copies
of $p_i$, each of which is selected with probability $\eta$.  We can
apply \theoref{eps:net} to $\PntSetA$ and show that if $\tukey(q,
\PntSetA) \geq \tfrac{d^2}{\eta} \ln (2d/\delta)$, then
$\probX{q,\PntSetA} \ge (1-\delta/2)$. Omitting the further details,
we conclude the following.

\begin{corollary}
    Let $\UPntSet=\brc{(p_1, \pr_1),\ldots,(p_n, \pr_n)}$ be a set of
    $n$ uncertain points in $\Re^d$ under the {\unipoint} model. For
    $1 \le i \le n$, set $w_i =
    \ceil{\frac{\ln(1-\pr_i)}{\ln(1-\delta/10n)}}$ be the weight of
    point $p_i$. If the (weighted) Tukey depth of a point $q \in
    \Re^d$ in $\brc{p_1, \ldots, p_n}$ is at least
    $\tfrac{10d^2n}{\delta} \ln (2d/\delta)$, then $\probX{q,
       \UPntSet} \ge 1-\delta$.
\end{corollary}

\mparagraph{Data structure.} %
Let $\UPntSet$ be a set of points in the uniform {\unipoint} model in
$\Re^2$, i.e., each point appears with probability $\pr$. We now
describe a data structure to estimate $\probX{q}$ for a query point $q
\in \Re^2$, within additive error $1/n$. We fix a parameter $\tDepth_0
= \tfrac{c}{\pr}\ln n$ for some constant $c>0$. Let
$\TT=\brc{x\in\Re^2 \mid \tukey(x, \UPntSet) \ge \tDepth_0}$ be the
set of all points whose Tukey depth in $P$ is at least $\tDepth_0$.
$\TT$ is a convex polygon with $O(n)$ vertices~\cite{m-ccpps-91}. By
\theoref{eps:net}, $\probX{q} \ge 1-1/n^2$ for all points $q \in \TT$,
provided that the constant $c$ is chosen appropriately. We also
preprocess $P$ for halfspace range reporting
queries~\cite{cgl-pogd-85}. $\TT$ can be computed in time $O(n\log^3
n)$~\cite{m-ccpps-91}, and constructing the half-plane range reporting
data structure takes $O(n\log n)$ time~\cite{cgl-pogd-85}. So the
total preprocessing time is $O(n\log^3 n)$, and the size of the data
structure is linear.

\parpic[r]{
\begin{picture}(0,0)%
\includegraphics{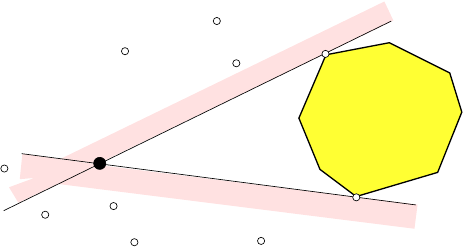}%
\end{picture}%
\setlength{\unitlength}{1657sp}%
\begingroup\makeatletter\ifx\SetFigFont\undefined%
\gdef\SetFigFont#1#2#3#4#5{%
  \fontfamily{#3}\fontseries{#4}\fontshape{#5}%
  \selectfont}%
\fi\endgroup%
\begin{picture}(5299,2801)(990,-3744)
\put(5191,-2294){\makebox(0,0)[lb]{\smash{{\SetFigFont{8}{9.6}{\rmdefault}{\mddefault}{\updefault}{\color[rgb]{0,0,0}$\TT$}%
}}}}
\put(3490,-2906){\makebox(0,0)[lb]{\smash{{\SetFigFont{8}{9.6}{\rmdefault}{\mddefault}{\updefault}{\color[rgb]{0,0,0}$\ell_2$}%
}}}}
\put(2058,-2638){\makebox(0,0)[lb]{\smash{{\SetFigFont{8}{9.6}{\rmdefault}{\mddefault}{\updefault}{\color[rgb]{0,0,0}$q$}%
}}}}
\put(4755,-1696){\makebox(0,0)[lb]{\smash{{\SetFigFont{8}{9.6}{\rmdefault}{\mddefault}{\updefault}{\color[rgb]{0,0,0}$\xi_1$}%
}}}}
\put(5008,-2991){\makebox(0,0)[lb]{\smash{{\SetFigFont{8}{9.6}{\rmdefault}{\mddefault}{\updefault}{\color[rgb]{0,0,0}$\xi_2$}%
}}}}
\put(3368,-2317){\makebox(0,0)[lb]{\smash{{\SetFigFont{8}{9.6}{\rmdefault}{\mddefault}{\updefault}{\color[rgb]{0,0,0}$\ell_1$}%
}}}}
\end{picture}%
}
%
A query is answered as follows. Given a query point $q \in \Re^2$, we
first test in $O(\log n)$ time whether $q \in \TT$. If the answer is
yes, we simply return $1$ as $\probX{q}$.  If not, we compute in
$O(\log n)$ time the two tangents $\ell_1, \ell_2$ of $\TT$ from
$q$. For $i=1,2$, let $\xi_i = \ell_i \cap \TT$, and let $\ell_i^-$ be
the half-plane bounded by $\ell_i$ that does not contain $\TT$. Set
$\UPntSet_q = \UPntSet \cap (\ell_1^-\cup\ell_2^-)$ and $n_q =
|\UPntSet_q|$. Let $\RSample_q$ be the subset of $\UPntSet_q$ by
choosing each point with probability $\pr$.

By querying the half-plane range reporting data structure with each of
these two tangent lines, we compute the set $\UPntSet_q$ in time
$O(\log n + n_q)$. Let $\omega_q = \Prob{q \notin \CHX{\RSample_q \cup
      \TT}}$. We compute $\omega_q$, in $(n_q \log n_q)$ time, by
adapting the algorithm for computing $\probX{q}$ described in
\secref{mem}.

The correctness and efficiency of the algorithm follow from the
following lemma, whose proof is omitted from this version.
\begin{lemma}
    \lemlab{witness} \lemlab{witness-size}
    
    For any point $q \not\in \TT$,
    \begin{inparaenum}[(i)]
        \item $\cardin{\Prob{q\in\CHX{ \RSample_q\cup\TT } } -
           \probX{q}} \leq 1/n$;
        \item $n_q \le 4\tDepth_0=O(\pr^{-1}\log n)$.
    \end{inparaenum}
\end{lemma}

By \lemref{witness-size}, $n_q = O(\pr^{-1}\log n)$, so the query
takes $O(\pr^{-1}\log (n)\log\log n)$ time. We thus obtain the
following.

\begin{theorem}
    Let $\UPntSet$ be a set of $n$ uncertain points in $\Re^2$ in the
    {\unipoint} model, where each point appears with probability
    $\pr$. $\UPntSet$ can be preprocessed in $O(n\log^3 n)$ time into
    a linear-size data structure so that for any point $q \in \Re^2$,
    returns a value $\tprobX{q}$ in $O(\pr^{-1}\log (n) \log\log n)$
    time such that $|\tprobX{q} - \probX{q}| \le 1/n$.
\end{theorem}

%
%

\section{$\some$-Hull}
\seclab{some_hull}

In this section, we consider the {\multipoint} model, i.e., $\UPntSet$
is a set of $m$ uncertain point defined by the pairs $\MultipointSet$.
A convex set $C
\subseteq \Re^2$ is called \textit{$\some$-dense} with respect to
$\UPntSet$ if it contains $\some$-fraction of each $(P_i, \PR_i)$,
i.e., $\sum_{p_i^j} \pr_i^j \geq \some$ for all $i \leq m$. The
\textit{$\some$-hull} of $\UPntSet$, denoted by $\CHSomeX{\UPntSet}$,
is the intersection of all convex $\some$-dense sets with respect to
$\UPntSet$. Note that for $m = 1$, $\CHSomeX{\UPntSet}$ is the set of
points whose Tukey depth is at least $1 - \some$. We first prove an
$O(n)$ upper bound on the complexity of $\CHSomeX{\UPntSet}$ and then
describe an algorithm for computing it.

\begin{theorem}
    \thmlab{linear_complexity}%
    Let $\UPntSet = \MultipointSet$ be a set of $m$ uncertain points
    in $\Re^2$ under the \multipoint model with $P = \bigcup_{i = 1}^m
    P_i$ and $|P| = n$. For any $\some \in [0, 1]$,
    $\CHSomeX{\UPntSet}$ has $O(n)$ vertices.
\end{theorem}

\begin{proof}
    We call a convex $\some$-dense set $C$ \textit{minimal} if there
    is no convex $\some$-dense set $C'$ such that $C' \subset C$. A
    minimal convex $\some$-dense set $C$ is the convex hull of $P \cap
    C$. Therefore $C$ is a convex polygon whose vertices are a subset
    of $P$. Obviously $\ConvexTauHull{\UPntSet}$ is the intersection of
    minimal convex $\some$-dense sets. Therefore each edge of
    $\ConvexTauHull{\UPntSet}$ lies on a line passing through a pair of
    points of $P$, i.e.,$\ConvexTauHull{\UPntSet}$ is the intersection of a set
    $H$ of halfplanes whose bounding line passes through a pair of
    points of $P$. Next we argue that $|H| \leq 2n$.
    
    Fix a point $p \in P$. We claim that $H$ contains at most two
    halfplanes whose bounding lines pass through $p$. Indeed if $p \in
    \Int(\ConvexTauHull{\UPntSet})$, then no bounding line of $H$
    passes through $p$; if $p \in \partial (\ConvexTauHull{\UPntSet})$,
    then at most two bounding lines of $H$ pass through $p$; and if $p
    \notin \ConvexTauHull{\UPntSet}$, then there are two tangent
    to$\ConvexTauHull{\UPntSet}$ from $p$. Hence at most two bounding
    lines of $H$ pass through $p$, as claimed.
\end{proof}

We describe a property of the set of lines supporting the edges of
$\CHSomeX{\UPntSet}$, which will be useful for computing
$\ConvexTauHull{\UPntSet}$. We call a line $\ell$ passing through a
point $p \in P_i$ \textit{$\some$-tangent} of $P_i$ at $p$ if one of
the open half-planes bounded by $\ell$ contains less than
$\some$-fraction of points of $P_i$ but the corresponding closed
half-plane contains at least $\some$-fraction of points. Using a
simple perturbation argument, the following can be proved.

\begin{lemma}
    \lemlab{tangent}%
    A line supporting an edge of $\CHSomeX{\UPntSet}$ is
    $\some$-tangent at two points of $P$.
\end{lemma}

\InNotESAVer{%
   \mparagraph{Algorithm.} %
We describe the algorithm for computing the upper boundary
$\UpperHull$ of$\ConvexTauHull{\UPntSet}$. The lower boundary of
$\ConvexTauHull{\UPntSet}$ can be computed analogously. It will be
easier to compute $\UpperHull$ in the dual plane. Let $\UpperHull^*$
denote the dual of $\UpperHull$.

Recall that the dual of a point $p = (a, b)$ is the line $p^*: y = ax
- b$, and the dual of a line $\ell: y = m x + c$ is the point $\ell^*
= (m, -c)$. The point $p$ lies above/below/on the line $\ell$ if and
only if the dual point $\ell^*$ lies above/below/on the dual line
$p^*$. Set $P_i^* = \brc{p_i^{j*} \mid p_i^j \in P_i}$ and $P^* =
\bigcup_{i = 1}^m P_i^*$. For a point $q \in \Re^2$ and for $i \leq
m$, let $\kappa(q, i) = \sum \pr_i^j$ where the summation is taken
over all points $p_i^j \in P_i$ such that $q$ lies below the dual line
$p_i^{j*}$. We define the $\some$-level $\TauLevel_i$ of $P_i^*$ to be
the upper boundary of the region $\brc{ q \in \Re^2 \mid \kappa(q, i)
   \geq \some}$. $\TauLevel_i$ is an $x$-monotone polygonal chain
composed of the edges of the arrangement
$\Arrangement(P_i^*)$. Further, the dual line of a point on
$\TauLevel_i$ is a $\some$-tangent line of $P_i$. Let $\TauLevel$ be
the lower envelope of $\TauLevel_1, \ldots, \TauLevel_m$.

Let $\ell$ be the line supporting an edge of $\UpperHull$. Using
\lemref{tangent}, it can be argued that the dual point $\ell^*$ is a
vertex of $\TauLevel$. Next, let $q$ be a vertex of $\UpperHull$, then
$q$ cannot lie above any $\some$-tangent line of any $P_i$, which
implies that the dual line $q^*$ passes through a pair of vertices of
$\TauLevel$ and does not lie below any vertex of $\TauLevel$.  Hence,
each vertex of $\UpperHull$ corresponds to an edge of the upper
boundary of the convex hull of $\TauLevel$.  By
\thmref{linear_complexity}, $\UpperHull^*$, the dual of $\UpperHull$,
has $O(n)$ vertices.

We now describe the algorithm for computing $\UpperHull^*$, which is
similar to the one used for computing the convex hull of a level in an
arrangement of lines~\cite{asw-actp-08, m-ccpps-91}. We begin by
describing a simpler procedure, which will be used as a subroutine in
the overall algorithm.

\begin{lemma}
    \lemlab{intersection_detection}%
    Given a line $\ell$, the intersection points of $\ell$ and
    $\TauLevel$ can be computed in $O(n \log n)$ time.
\end{lemma}

\begin{proof}
    We sort the intersections of the lines of $P^*$ with $\ell$. Let
    $\langle q_1, \ldots, q_u\rangle, \; u \leq n$, be the sequence of
    these intersection points. For every $i \leq m$, $\kappa(q_1, i)$
    can be computed in a total of $O(n)$ time. Given
    $\brc{\kappa(q_{j-1}, i) \mid 1\leq i \leq m}$,
    $\brc{\kappa(q_{j}, i) \mid 1\leq i \leq m}$ can be computed in
    $O(1)$ time. A point $q_j \in \TauLevel$ if $q_j \in \TauLevel_i$
    for some $i$ and lies below $\TauLevel_i'$ for all other
    $i'$. This completes the proof of the lemma.
\end{proof}

The following two procedures can be developed by plugging
\lemref{intersection_detection} into the parametric-search
technique~\cite{asw-actp-08, m-ccpps-91}. \smallskip
\begin{compactenum}[\qquad(A)]
    \item \itemlab{subroutine:previous} Given a point $q$, determine
    whether $q$ lies above $\UpperHull^*$ or return the tangent lines
    of $\UpperHull^*$ from $q$. This can be done in $O(n \log^2 n)$
    time.
    
    \item \itemlab{subroutine} %
    Given a line $\ell$, compute the edges of $\UpperHull^*$ that
    intersect $\ell$, in $O(n \log^3 n)$ time. (Procedure
    \itemref{subroutine} uses \itemref{subroutine:previous} and
    parametric search.)
\end{compactenum}
\smallskip

Given \itemref{subroutine}, we can now compute $\UpperHull^*$ as
follows. We fix a parameter $r > 1$ and compute a $(1/r)$-cutting%
\footnote{A \textit{$(1/r)$-cutting} of $P^*$ is a triangulation $\Xi$
   of $\Re^2$ such that each triangle of $\Xi$ crosses at most $n/r$
   lines of $P^*$.} %
$\Xi = \brc{\Delta_1, \ldots, \Delta_u}$, where $u = O(r^2)$.  For
each $\Delta_i$, we do the following. Using \itemref{subroutine} we
compute the edges of $\UpperHull^*$ that intersect $\partial
\Delta_i$. We can then deduce whether $\Delta_i$ contains any vertex
of $\UpperHull^*$. If the answer is yes, we solve the problem
recursively in $\Delta_i$ with the subset of lines of $P^*$ that cross
$\Delta_i$. We omit the details from here and conclude the following.

}

\begin{theorem}
    Given a set $\UPntSet$ of uncertain points in $\Re^2$ under the
    {\multipoint} model with a total of $n$ sites, and a parameter
    $\some \in [0, 1]$, the $\some$-hull of $\UPntSet$ can be computed
    in $O(n\log^3 n)$ time.
\end{theorem}

{\small
   \noindent\textbf{Acknowledgments.} P. Agarwal and W. Zhang are
   supported by NSF under grants CCF-09-40671, CCF-10-12254, and
   CCF-11-61359, by A{R}O grants W911NF-07-1-0376 and
   W911NF-08-1-0452, and by an E{RD}C contract W9132V-11-C-0003.
   S.~Har-Peled is supported by NSF grants CCF-09-15984 and
   CCF-12-17462.
   S. Suri and H. Y{\i}l{}d{\i}z are supported by NSF grants
   CCF-1161495 and CNS-1035917.
}

\InESAVer{%
\newcommand{\etalchar}[1]{$^{#1}$}
 \providecommand{\CNFX}[1]{ {\em{\textrm{(#1)}}}}
  \providecommand{\CNFSoCG}{\CNFX{SoCG}}
  \providecommand{\CNFESA}{\CNFX{ESA}}
  \providecommand{\CNFPODS}{\CNFX{PODS}}
  \providecommand{\CNFWADS}{\CNFX{WADS}}

}
\InNotESAVer{%
\newcommand{\etalchar}[1]{$^{#1}$}
 \providecommand{\CNFX}[1]{ {\em{\textrm{(#1)}}}}
  \providecommand{\CNFSoCG}{\CNFX{SoCG}}
  \providecommand{\CNFESA}{\CNFX{ESA}}
  \providecommand{\CNFPODS}{\CNFX{PODS}}
  \providecommand{\CNFWADS}{\CNFX{WADS}}

}


%
%


\InNotESAVer{%
\appendix

\section{Proof of \lemref{projection}}
\apndlab{escape:unique}

\section{Computing Face Probabilities in Radial Order}
\apndlab{radial:D:D}

\begin{figure}[t]
    \centering 
\definecolor{a0a0a4}{RGB}{160,160,164}
\definecolor{000000}{RGB}{0,0,0}
\definecolor{808080}{RGB}{128,128,128}
\begin{tikzpicture}[x=1cm,y=1cm]
\clip (15,0) rectangle (18,3);
\node [color=000000,inner sep=0.5pt,above,shift={(0,3pt)}] at (0.49296,0.00762142) {$\circled{$\GroupX{2}$}$};
\node [color=000000,inner sep=0.5pt,below right,shift={(3pt,-3pt)}] at (1.37078,0.945937) {$\mpt$};
\node [color=000000,inner sep=0.5pt,above,shift={(0,3pt)}] at (2.45777,0.885647) {$\circled{$\GroupX{1}$}$};
\node [color=000000,inner sep=0.5pt,right,shift={(4pt,0)}] at (1.92125,2.11751) {$\siX{i}$};
\node [color=000000,inner sep=0.5pt,above left,shift={(-3pt,3pt)}] at (6.06623,2.36556) {$\siX{2}$};
\node [color=000000,inner sep=0.5pt,above left,shift={(-3pt,3pt)}] at (5.51536,1.86844) {$\siX{1}$};
\node [color=000000,inner sep=0.5pt,above right,shift={(3pt,3pt)}] at (6.60369,2.00193) {$...$};
\node [color=000000,inner sep=0.5pt,right,shift={(3pt,0)}] at (11.4742,1.51824) {$\mpt$};
\node [color=000000,inner sep=0.5pt,below right,shift={(3pt,-3pt)}] at (6.49617,1.36459) {$\mpt$};
\node [color=000000,inner sep=0.5pt,below right,shift={(3pt,-3pt)}] at (6.28795,0.785994) {$...$};
\node [color=000000,inner sep=0.5pt,above,shift={(0,3pt)}] at (11.4204,0.107486) {$\circled{$\SGroup'$}$};
\node [color=000000,inner sep=0.5pt,below left ,shift={(-3pt,-3pt)}] at (5.51539,0.960659) {$\siX{n}$};
\node [color=000000,inner sep=0.5pt,left,shift={(-3pt,0)}] at (10.8293,2.37814) {$\siX{i-1}$};
\node [color=000000,inner sep=0.5pt,right,shift={(3pt,0)}] at (12.0385,2.37814) {$\siX{i}$};
\node [color=000000,inner sep=0.5pt,right,shift={(3pt,0)}] at (16.3098,1.23982) {$\mpt$};
\node [color=000000,inner sep=0.5pt,above left,shift={(-3pt,3pt)}] at (16.6013,1.84811) {$\f_j$};
\node [color=000000,inner sep=0.5pt,above,shift={(0,3pt)}] at
(15.4861,1.92414) {$\siX{i}$};
\node [color=000000,inner sep=0.5pt,above,shift={(0,3pt)}] at
(17.1462,0.568163) {$\circled{$\GroupX{1}$}$};
\node [color=000000,inner sep=0.5pt,right,shift={(3pt,0)}] at
(16.8927,2.4564) {$\siX{j}$};
\draw[color=a0a0a4,line width=1,dotted] (-2.14531,0.945937) -- (1.37078,0.945937);
\draw[color=a0a0a4,line width=1,dotted] (0.637247,-0.615251) -- (2.56193,3.48108);
\draw[color=a0a0a4,line width=1,dotted] (15.3266,-0.812143) -- (17.5569,3.84246);
\draw[color=a0a0a4,line width=1,dotted] (4.7697,1.36459) -- (8.49141,1.36459);
\draw[color=000000,line width=1] (1.92125,2.11751) -- (1.37078,0.945937);
\draw[color=808080,line width=1,densely dashed] (6.49617,1.36459) -- (6.28795,0.785994);
\draw[color=a0a0a4,line width=1,dotted] (12.8309,-0.290739) -- (10.1622,3.2677);
\draw[color=808080,line width=1,densely dashed] (6.49617,1.36459) -- (6.60369,2.00193);
\draw[color=808080,line width=1,densely dashed] (6.49617,1.36459) -- (5.51539,0.960659);
\draw[color=a0a0a4,line width=1,dotted] (12.6337,3.28505) -- (10.3188,-0.24235);
\draw[color=808080,line width=1,densely dashed] (6.06623,2.36556) -- (6.49617,1.36459);
\draw[color=808080,line width=1,densely dashed] (5.51536,1.86844) -- (6.49617,1.36459);
\draw[color=000000,line width=1] (16.3098,1.23982) -- (16.8927,2.4564);
\path [fill=000000] (16.8927,2.4564) circle (1.5pt);
\path [fill=000000] (6.06623,2.36556) circle (1.5pt);
\path [fill=000000] (6.49617,1.36459) circle (1.5pt);
\path [fill=000000] (5.51536,1.86844) circle (1.5pt);
\path [fill=000000] (1.37078,0.945937) circle (1.5pt);
\path [fill=000000] (1.92125,2.11751) circle (1.5pt);
\path [fill=000000] (5.51539,0.960659) circle (1.5pt);
\path [fill=000000] (12.0385,2.37814) circle (1.5pt);
\path [fill=000000] (11.4742,1.51824) circle (1.5pt);
\path [fill=000000] (10.8293,2.37814) circle (1.5pt);
\path [fill=000000] (16.3098,1.23982) circle (1.5pt);
\path [fill=000000] (15.4861,1.92414) circle (1.5pt);
\end{tikzpicture}

    \caption{}
    \figlab{ocomp}
\end{figure}

Similar to the planar case, we can improve the computation time for
the $i$\th term to $O(n^{d-1})$ by considering the facets $\f(\SFACE)$
in radial order. In particular, let $\SRIDGE \subseteq \Sites$ be a
subset of $(d-2)$ sites. Let $\f_j$ denote the $(d-1)$-dimensional
simplex $\f(\SRIDGE \cup \brc{\query} \cup \brc{\siX{j}})$ where
$\siX{j} \not \in \SRIDGE$ and $\siX{j} \not = \siX{i}$. We can
compute the probability that $\f_j$ is a facet of $\CC$ for all facets
$\f_j$ in constant amortized time as follows. We project all sites to
the two-dimensional plane passing through $\query$ and orthogonal to
the $(d-2)$-dimensional hyperplane defined by $\SRIDGE \cup
\brc{\query}$. (Such a plane is known as an orthogonal complement.)
The hyperplane defined by $\SRIDGE \cup \brc{\query}$ projects onto
$\query$ on this plane. Moreover, each facet $\f_j$ projects to a line
segment extending from $\query$. When we need to compute the
probability that $\f_j$ is a facet of $\CC$, the set $\GroupX{1}$
includes the sites on the other side of the line supporting $\f_j$'s
projection with respect to $\siX{i}$. (See \figref{ocomp}.) We compute
probabilities for the facets $\f_j$ based on their radial order around
$\query$. The probability for the next facet in the sweep can be
computed by modifying the probability of the previous facet in
constant amortized time as we have done for the planar case, as we can
efficiently track how $\GroupX{1}$ changes. As a final note, we point
out that the total cost of all sorting involved is $O(n^{d-1}\log{n})$
which is less than the overall cost of $O(n^d)$.

\section{Membership Probability Algorithms %
   in the {\Multipoint} Model}
\apndlab{memmulti}

\subsection{The Planar Case}

Let $\UPntSet$ be an uncertain set of points in the {\multipoint}
model defined by site groups $\brc{\SiX{1},\ldots,\SiX{m}}$. We denote
the $j$\th site in $\siX{i}$ by $\sijY{i}{j}$ and its probability by
$\prijY{i}{j}$. For simplicity, we set $n_i=|\SiX{i}|$.  We define
$\Sites$ to be the set all sites, i.e. $\Sites = \bigcup_{1 \le i \le
   m}{\siX{i}}$, and set $n=\cardin{\Sites} = \sum_{1 \le i \le
   m}{n_i}$. Under this setting, we want to compute the membership
probability of a given point $\query$. Recall that the sites from a
single site group $\siX{i}$ are dependent, i.e., they cannot co-exist
in an outcome $\SetB \subseteq \Sites$ of the probabilistic
experiment.

The algorithm for the {\unipoint} model easily extends to the
{\multipoint} model. The main difference is the way we compute the
probability for an edge. The rest of the algorithm remains the mostly
the same. We now see this in more depth.

Let $\VV$ and $\CC$ be defined as before. As in the {\unipoint} model,
$\query$ is in the convex hull of $\SetB$ if and only if $\query \in
\VV$. We follow a similar strategy and decompose $\Prob{\query \in
   \VV}$ as follows:
\begin{align*}
    \Prob{\query \in \VV}%
    =%
    \Prob{\SetB = \emptyset} + \sum_{\substack{1 \le i \le m \\ 1 \leq
          j \le n_i}} \Prob{\query\sijY{i}{j} \text{ is the witness
          edge of } \query \notin \CHX{\SetB}}.
\end{align*}
The first term is trivial to compute in $O(n)$ time. We compute the
probability that $\query\sijY{i}{j}$ forms a witness edge of $\SetB$ as
follows. Let $\GroupX{i,j}$ be the set of sites to the right of the
line $\Line{\query\sijY{i}{j}}$ where the right direction is with
respect to the vector $\vect{\query\sijY{i}{j}}$. As in the {\unipoint}
model, the segment $\query\sijY{i}{j}$ is the witness edge of $\SetB$ if
and only if $\sijY{i}{j} \in \SetB$ and $\SetB \cap \GroupX{ i, j } =
\emptyset$. We can write the corresponding probability as follows:
\begin{align*}
    \Prob{\sijY{i}{j} \in \SetB \;\land\; \SetB \cap \GroupX{{i,j}}%
       =%
       \emptyset}%
    \hspace{-2cm} &\hspace{2cm}%
    =%
    \Prob{\sijY{i}{j} \in \SetB} \times \Prob{\SetB \cap
       \GroupX{{i,j}}%
       = \emptyset \st \sijY{i}{j} \in \SetB} \\
    &= %
    \Prob{\sijY{i}{j} \in \SetB} \times \prod_{1 \le k \le m}
    \Prob{\SetB \cap \GroupX{{i,j}} \cap \SiX{k}
       = \emptyset \st \sijY{i}{j} \in \SetB} \\
    &=%
    \Prob{\sijY{i}{j} \in \SetB} \times \prod_{\substack{1 \le k \le m \\
          k \not = i}} \Prob{\SetB \cap \SiX{k} \cap \GroupX{{i,j}} =
       \emptyset} \\
    &=%
    \prijY{i}{j} \times \prod_{\substack{1 \le k \le m \\ k \not = i}}
    \pth{ 1 - \sum_{l \st \sijY{k}{l} \in \GroupX{{i,j}}} \prijY{k}{l}
    }.
\end{align*}
This expression can be easily computed in $O(n)$ time. It follows that
one can compute $\inv{\mprX{\query}}$, thus $\mprX{\query}$, in
$O(n^2)$ time.

As before, the computation time can be improved to $O(n\log{n})$ by
computing the witness edge probabilities in radial order around
$\query$. Let the circular sequence $\spi(1),\spi(2),\ldots,\spi(n)$ be
the counter-clockwise order of all sites around $\query$, where each
$\spi(i)$ is a distinct site $\sijY{a}{b}$. We first compute the
probability that $\query\spi(1)$ is the witness edge in $O(n)$ time and
also remember the values of the intermediate factors used in the
computation. (The factors inside the $\prod_{1 \le k \le m}$
expression.) Then, for increasing values of $i$ from $2$ to $n$, we
compute the probability that $\query\spi(k)$ is the witness edge by
updating the probability for $\query\spi(k-1)$. As a first step to this
update, we update the values of the intermediate factors. To be more
specific, let $W_i$ denote the set of sites in the open wedge bounded
by the lines $\Line{\query\spi(i)}$ and $\Line{\query\spi(i-1)}$. Also,
for simplicity, assume that $\spi(k)=\sijY{a}{b}$ and
$\spi(k-1)=\sijY{c}{d}$. Notice that $\GroupX{{a,b}} = \GroupX{{c,d}}
\cup \brc{\sijY{c}{d}} \setminus W_i$. Then, for each site
$\sijY{k}{l}$ in $W_i$, the $k$\th factor increases by
$\prijY{k}{l}$. Also, the $c$\th factor decreases by
$\prijY{c}{d}$. Finally, we temporarily set the value of the $a$\th
factor to 1 (to cover the case $k \not = i$ in the expression). Then,
we can compute the witness edge probability for $\query\spi(k)$ by
multiplying the probability of $\query\spi(k-1)$ with
${\prijY{a}{b}}/{\prijY{c}{d}}$ and the multiplicative change in each
intermediate factor. The cost of a single update is $O(1)$ amortized,
as each site can contribute to at most 4 updates as in the {\unipoint}
case.

\subsection{The $d$-dimensional case}

All of the arguments in the $d$-dimensional algorithm are also easily
extended to the multipoint. As before, we compute $\mprX{\query}$ by
computing the probability $\Prob{\query \in \VV}$. Following the same
strategy, we decompose it as

\begin{flalign*}
    \Prob{\Bigl. \query \in \VV} & = \Prob{\Bigl. \query \altequiv \lowd(\SetBA)} \\
    & + \sum_{1 \le i \le m} \left( \sum_{1 \le j \le
           n_i}{\Prob{\Bigl. \sijY{i}{j} \altequiv \lowd(\SetBA)
              \;\land\; \query \in \VV}} \right).
\end{flalign*}

It is trivial to compute the first term in $O(n)$ time. We now show
how to compute each term inside the summations in $O(n^{d-1})$
time. This implies a total time of $O(n^d)$.

Clearly, \lemref{projection} extends to the {\multipoint} model, so we
can use $\sijY{i}{j}$-escaping facets to decompose our
probability. Given a subset of sites$\Sites_{\alpha} \subseteq \Sites
\!\setminus\!  \brc{\sijY{i}{j}}$ of size $(d-1)$, define $\f(\SFACE)$
to be the $(d-1)$-dimensional simplex whose vertices are the points in
$\Sites_{\alpha}$ and $\query$. Then,
\begin{align*}
    &%
    \Prob{\Bigl. \sijY{i}{j} \altequiv \lowd(\SetBA) \;\land\; \query
       \in \VV}%
    =%
    \Prob{\Bigl. \sijY{i}{j} \altequiv \lowd(\SetBA) \; \land\; \query'
       \in \VVP}%
    \\& \qquad%
    + \sum_{\substack{ \SFACE \subseteq {\Sites \setminus \brc{\sijY{i}{j}}} \\
          |\SFACE|=(d-1) \\ \f(\SFACE) \text{ is
          }\sijY{i}{j}\text{-escaping for $\query$}}}
    \Prob{\Bigl. \sijY{i}{j} \altequiv \lowd(\SetBA) \;\land\;
       \f(\SFACE) \text{ is a facet of } \CC}.
\end{align*}
The first term is computed recursively. We compute each term of the
summation as follows. Let $\FaceIndices$ be the set of group indices
of the sites in $\SFACE$, i.e., $\FaceIndices=\brc{\ii \st \exists \jj
   \,.\, \sijY{\ii}{\jj} \in \SFACE}$. As before, let $\GroupX{1}
\subseteq \Sites$ be the subset of sites which are on the other side
of the hyperplane supporting $\f(\SFACE)$ with respect to
$\sijY{i}{j}$. Let $\GroupX{2} \subseteq \Sites$ be the subset of
sites that are below $\sijY{i}{j}$ along the $x_d$-axis. Following the
same strategy, we write the desired probability as the probability
that all points in $\SFACE$ and $\sijY{i}{j}$ exist in $\SetB$, and
all points in $\GroupX{1} \cup \GroupX{2}$ are absent from
$\SetB$. This probability is clearly zero, if any of the following
conditions hold:

\begin{itemize}
    \item $\SFACE \cap \GroupX{2} \not = \emptyset$.
    \item $\sijY{i}{j}$ has a higher $x_d$-coordinate than $\query$.
    \item $\SFACE$ contains any two sites from the same uncertain
    point $\SiX{k}$.
    \item $\SFACE$ contains any site from $\siX{i}$.
\end{itemize}
Otherwise, we can write the probability as follows:
\begin{flalign*}
    & \Prob{\sijY{i}{j} \in \SetB \;\land\; \SFACE \cap \SetB = \SFACE
       \;\land\; \SetB \cap (\GroupX{1} \cup \GroupX{2}) = \emptyset}%
    \displaybreak[0]\\
    & \qquad = \Prob{\sijY{i}{j} \in \SetB} \,\times\, \Prob{\SFACE
       \cap \SetB = \SFACE \mid \sijY{i}{j} \in \SetB} \,\times\, %
    \displaybreak[0]\\
    & \qquad \qquad \Prob{\SetB \cap (\GroupX{1} \cup \GroupX{2}) =
       \emptyset \mid \sijY{i}{j} \in \SetB \;\land\; \SFACE \cap
       \SetB = \SFACE}%
    \displaybreak[0]\\
    & \qquad = \Prob{\sijY{i}{j} \in \SetB} \,\times\,
    \Prob{\Bigl. \SFACE
       \cap \SetB = \SFACE} \,\times\, \\
    & \qquad \qquad \Prob{\SetB \cap (\GroupX{1} \cup \GroupX{2}) =
       \emptyset \mid \sijY{i}{j} \in \SetB \;\land\; \SFACE \cap
       \SetB = \SFACE}%
    \displaybreak[0]\\
    & \qquad = \Prob{\sijY{i}{j} \in \SetB} \,\times\,
    \Prob{\Bigl. \SFACE \cap \SetB = \SFACE} \,\times\,
    \displaybreak[0]\\
    & \qquad \qquad \prod_{\substack{1 \le \ii \le m \\ \ii \not = i \\
          \ii \not \in \FaceIndices}} \bigg( \Prob{\SiX{\ii} \cap
       \SetB \cap (\GroupX{1} \cup \GroupX{2}) = \emptyset} \bigg)
    \displaybreak[0]\\
    & \qquad = \prijY{i}{j} \,\times\, \prod_{\ii,\jj \st
       \sijY{\ii}{\jj} \in \SFACE}{\prijY{\ii}{\jj}} \,\times\,
    \prod_{\substack{1 \le \ii \le m \\ \ii \not = i \\ \ii \not \in
          \FaceIndices}}{\left( 1- \sum_{\jj \st \sijY{\ii}{\jj} \in
              (\GroupX{1} \cup \GroupX{2})}{\prijY{\ii}{\jj}}
       \right)}.
\end{flalign*}

The expression takes linear time to compute and thus summation term
can be computed in $O(n^d)$ time. Then, by induction, the computation
of the term for the site $\sijY{i}{j}$ takes $O(n^d)$ time. As before,
we can improve the computation time each term to $O(n^{d-1})$ by
considering the facets $\f(\SFACE)$ in radial order. This implies a
total complexity of $O(n^d)$ for the algorithm.

\section{Computing the Probability Map in %
   $O(n^4)$ Time}
\apndlab{improve}

In this section, we describe how to compute the \emphi{probability
   map} for a given set of uncertain points on the plane in $O(n^4)$
time, rather than $O(n^5\log{n})$. For simplicity, we show how to
compute the probability associated with each \emphi{face} (or cell) of
the probability map. Our algorithm can easily be adapted to compute
the probabilities of edges and vertices as well, by taking care of
degeneracies. Also, we assume that the input is given in the
{\unipoint} model, however, we briefly explain how to extend the
algorithm to the {\multipoint} model.

The high level idea of our algorithm is as follows. Recall that the
structure of the probability map is an arrangement of $O(n^2)$ lines,
containing $O(n^4)$ faces. We first compute the membership probability
of one of the faces, say $F$, in $O(n\log{n})$ time. We then compute
the membership probabilities of the faces neighboring $F$, in $O(1)$
time per each, by modifying the probability of $F$. By repeatedly
applying the same idea of expanding into the neighbors, we can compute
the probability of all faces in $O(n^4)$ time.

To complete our algorithm, we now describe how we can compute the
probability of a face $F'$ by using the already computed probability
of one of its neighbors $F$.
Without loss of generality, assume that $F$ and $F'$ are separated by
a vertical line passing through the sites $\siX{i}$ and $\siX{j}$ and
$F$ is to the left of $F'$. Notice that the boundary separating $F$
and $F'$ is only a segment of the vertical line and does not contain
$\siX{i}$ or $\siX{j}$. Now imagine that a point $\query$ moves through
this boundary, crossing from $F$ to $F'$. It is easy to see that the
change in the membership probability of $\query$ is due to the changes
in witness edge probabilities of the segments $\query\siX{i}$ and
$\query\siX{j}$, as other sites are irrelevant. We now describe the
change in the witness edge probability of $\query\siX{i}$. The
probability of $\query\siX{j}$ changes analogously. The change in the
probability of $\query\siX{i}$ happens differently for two cases:

\begin{enumerate}
    \item {\bf $\siX{i}$ is above $\siX{j}$:} Then, $\siX{j}$ switches
    from the right side of the line $\line{\query\siX{i}}$ to its left
    side (where right direction is with respect to the vector
    $\vec{\query\siX{i}}$). Consequently, the probability of
    $\query\siX{i}$ changes by a factor of $\frac{1}{\inv{\priX{j}}}$.
    \item {\bf $\siX{i}$ is below $\siX{j}$:} Then, $\siX{j}$ switches
    from the left side of the line $\line{\query\siX{i}}$ to its right
    side. Consequently, the probability of $\query\siX{i}$ changes by
    a factor of $\inv{\priX{j}}$.
\end{enumerate}

The changes clearly require constant time operations, and thus the
membership probability of $F'$ can be computed in $O(1)$ time.

The extension of this technique to the {\multipoint} model is
straightforward. The only major difference is that we need to remember
\InNotESAVer{(similar to what is done in \apndref{memmulti})} the
intermediate factors when switching from face to face, as updating the
witness edge probabilities requires updating these factors first. The
total cost of an update remains $O(1)$ because each face switch
updates one intermediate factor of two witness edge probabilities.

}
\end{document}